\documentclass[sigconf, nonacm]{acmart}

\newcommand\vldbdoi{XX.XX/XXX.XX}
\newcommand\vldbpages{XXX-XXX}
\newcommand\vldbvolume{17}
\newcommand\vldbissue{12}
\newcommand\vldbyear{2024}
\newcommand\vldbauthors{\authors}
\newcommand\vldbtitle{\shorttitle} 
\newcommand\vldbavailabilityurl{}
\newcommand\vldbpagestyle{plain}

\newif\ifcomments  
\commentsfalse

\usepackage{amstext}
\usepackage{amsthm}
\usepackage{multirow}
\usepackage[skip=0pt]{subcaption}
\usepackage[shortlabels]{enumitem}
\usepackage{cancel}
\usepackage{wrapfig}
\usepackage{hyperref}
\usepackage{mathtools}
\usepackage{xcolor}

\usepackage{comment}
\usepackage{todonotes}
\usepackage[normalem]{ulem}

\usepackage{algorithm,algorithmic}

\hypersetup{final}

\newcommand{\calA}{\ensuremath{\mathcal{A}}}

\newcommand{\calN}{\ensuremath{\mathcal{N}}}
\newcommand{\calO}{\ensuremath{\mathcal{O}}}

\newcommand{\calS}{\ensuremath{\mathcal{S}}}
\newcommand{\calT}{\ensuremath{\mathcal{T}}}

\renewcommand{\Pr}{\mathop{\mathbf{Pr}}}

\newtheorem{lem}{Lemma}[section]

\newtheorem{thm}[lem]{Theorem}

\newtheorem{defn}[lem]{Definition}

\newtheorem{definition}[lem]{Definition}

\makeatletter
\newcommand{\vast}{\bBigg@{4}}
\newcommand{\Vast}{\bBigg@{5}}
\makeatother

\DeclarePairedDelimiterX{\infdivx}[2]{(}{)}{%
  #1\;\delimsize\|\;#2%
}

\newcommand{\mypar}[1]{\smallskip
	\noindent{\textbf{{#1}:}}}
	
\renewcommand{\epsilon}{\varepsilon}

\newcommand{\tree}{\calT}
\newcommand{\init}{{\texttt{InitializeTree}}\,}
\newcommand{\addt}{{\texttt{AddToTree}}\,}

\newcommand{\gettotal}{{\texttt{GetTotalSum}}\,}

\newcommand{\node}{\texttt{node}}

\newcommand{\nd}{\texttt{node}}
\newcommand{\tr}{\texttt{Tr}}

\newcommand{\countt}{\texttt{count}}

\newcommand{\pr}[2]{{\ifx&#1& \mathbb{P} \else \underset{#1}{\mathbb{P}} \fi \left[#2\right]}}
\newcommand{\thresh}{\mu}

\newcommand{\DpTree}{\texttt{DP-Tree}\,}
\newcommand{\DpTrees}{\texttt{DP-Trees}\,}
\newcommand{\Sum}{\texttt{Sum}}

\newcommand{\groupby}{\texttt{GROUP BY}\,}
\newcommand{\bigOh}[1]{\calO\left(#1\right)}
\newcommand{\lrt}{\texttt{LRT}}

\newcommand{\recordkey}{\texttt{key}\,}
\newcommand{\recordvalue}{\texttt{value}\,}
\newcommand{\recordtimestamp}{\texttt{timestamp}\,}
\newcommand{\recorduserid}{\texttt{user\_id}\,}
\newcommand{\level}{\texttt{level}}
\usepackage{fullpage}  
\usepackage{graphicx} 
\usepackage{fancyvrb} 
\usepackage{tgcursor}  
\usepackage{listings}
\usepackage{amsfonts} 
\usepackage{float}

\begin{document}

\title{Differentially Private Stream Processing at Scale}

\author{Bing Zhang$^{1}$, Vadym Doroshenko$^{\dagger1}$, Peter Kairouz$^{\dagger3}$, Thomas Steinke$^{\dagger2}$, Abhradeep Thakurta$^{\dagger2}$, Ziyin Ma$^1$, Eidan Cohen$^1$, Himani Apte$^1$, Jodi Spacek$^1$}
\def \authors{Bing Zhang, Vadym Doroshenko, Peter Kairouz, Thomas Steinke, Abhradeep Thakurta, Ziyin Ma, Eidan Cohen, Himani Apte, Jodi Spacek}
\affiliation{\institution{$^1$Google}\country{}}
\affiliation{\institution{$^2$Google DeepMind}\country{}}
\affiliation{\institution{$^3$Google Research}\country{}}
\email{{zhangbing,dvadym,kairouz,steinke,athakurta,ziyinma,eidanch,himaniapte,jodes}@google.com}

\begin{abstract}
We design, to the best of our knowledge, the first differentially private (DP) stream aggregation processing system at scale. Our system -- \emph{Differential Privacy SQL Pipelines (DP-SQLP)} -- is built using a streaming framework similar to Spark streaming, and is built on top of the Spanner database and the F1 query engine from Google.

Towards designing DP-SQLP we make both algorithmic and systemic advances, namely, we (i) design a novel (user-level) DP key selection algorithm that can operate on an unbounded set of possible keys, and can scale to one billion keys that users have contributed, (ii) design a preemptive execution scheme for DP key selection that avoids enumerating all the keys at each triggering time, and (iii) use algorithmic techniques from DP continual observation to release a continual DP histogram of user contributions to different keys over the stream length. We empirically demonstrate the efficacy by obtaining at least $16\times$ reduction in error over meaningful baselines we consider. We implemented a streaming differentially private user impressions for Google Shopping with DP-SQLP. The streaming DP algorithms are further applied to Google Trends.
\end{abstract}

\maketitle
\def\thefootnote{$\dagger$}\footnotetext{Contributed equally.}

%%% do not modify the following VLDB block %%
%%% VLDB block start %%%
\pagestyle{\vldbpagestyle}
\begingroup\small\noindent\raggedright\textbf{PVLDB Reference Format:}\\
\vldbauthors. \vldbtitle. PVLDB, \vldbvolume(\vldbissue): \vldbpages, \vldbyear.\\
\href{https://doi.org/\vldbdoi}{doi:\vldbdoi}
\endgroup
\begingroup
\renewcommand\thefootnote{}\footnote{\noindent
This work is licensed under the Creative Commons BY-NC-ND 4.0 International License. Visit \url{https://creativecommons.org/licenses/by-nc-nd/4.0/} to view a copy of this license. For any use beyond those covered by this license, obtain permission by emailing \href{mailto:info@vldb.org}{info@vldb.org}. Copyright is held by the owner/author(s). Publication rights licensed to the VLDB Endowment. \\
\raggedright Proceedings of the VLDB Endowment, Vol. \vldbvolume, No. \vldbissue\ %
ISSN 2150-8097. \\
\href{https://doi.org/\vldbdoi}{doi:\vldbdoi} \\
}\addtocounter{footnote}{-1}\endgroup
%%% VLDB block end %%%

%%% do not modify the following VLDB block %%
%%% VLDB block start %%%
\ifdefempty{\vldbavailabilityurl}{}{
\vspace{.3cm}
\begingroup\small\noindent\raggedright\textbf{PVLDB Artifact Availability:}\\
The source code, data, and/or other artifacts have been made available at \url{\vldbavailabilityurl}.
\endgroup
}
%%% VLDB block end %%%

\section{Introduction}
\label{sec:intro}

Analysis of streaming data with differential privacy (DP)~\cite{DMNS} has been studied from the initial days of the field~\cite{Dwork-continual,CSS11-continual}, and this has been followed up in a sequence of works that include computing simple statistics~\cite{perrier2018private}, to machine learning applications (a.k.a.~online learning)~\cite{JKT-online,thakurta2013nearly,agarwal2017price,kairouz2021practical,erlingsson2019private}. While all of these works focus on the abstract algorithmic design for various artifacts of streaming data processing, to the best of our knowledge, none of them focus on designing a scalable stream processing system. In this work, we primarily focus on designing a scalable DP stream processing system, called Differential Privacy SQL Pipelines (DP-SQLP), and make algorithmic advances along the way to cater to the scalability needs of it. DP-SQLP is implemented using a streaming framework similar to Spark streaming~\cite{Zaharia13streamingspark}, and is built on top of the Spanner database~\cite{corbett2013spanner} and F1 query engine~\cite{samwel2018f1} from Google. We also present production applications with two use cases in Section~\ref{sec:application}. The first is a real world use case that deploys DP-SQLP in~\emph{Google Shopping} to generate streaming page-view counts. The second applies the streaming DP algorithm to~\emph{Google Trends}.

In this paper we consider a data stream to be an unbounded sequence of tuples of the form $(\recordkey,$ $\recordvalue,$ $\recordtimestamp,$ $\recorduserid)$ that gets generated continuously in time. We also have a discrete set of times (a.k.a.~\emph{triggering times}) $\tr = [t^{\sf tr}_1, t^{\sf tr}_2,...,t^{\sf tr}_T]$. The objective is to output the sum of all the values for each of the keys at each time $t^{\sf tr}_i$, while preserving $(\epsilon,\delta)$-DP~\cite{DMNS} over the entire output stream with respect to all of the contributions with the same $\recorduserid$. Although most prior research has extended simple one-shot DP algorithms to the streaming setting \cite{Dwork-continual,CSS11-continual,cardoso2022differentially}, designing a scalable DP-streaming system using off-the-shelf algorithms is challenging because of the following reasons\footnote{Our work is most closely related to~\cite{cardoso2022differentially}. We defer a full comparison to Section~\ref{sec:related}.}:

\begin{enumerate}
    \item \mypar{Unknown key space} A data stream processing system can only process the data that has already arrived. For example, keys for a $\groupby$ operation are not known in advance; instead we discover new keys as they arrive. To ensure $(\epsilon,\delta)$-DP one has to ensure the set of keys for which the statistics are computed is~\emph{stable} to change of an individual user's data. That is, we can only report statistics for a particular $\recordkey$ when enough users have contributed to it; to ensure DP the threshold for reporting a key must be randomized.
    
    \item \mypar{Synchronous execution} The execution of the streaming system is driven by the data stream. That is, we must process the data as it arrives, and cannot run asynchronously at times when there is nothing to trigger execution.  We refer to the times when our system runs as \emph{triggering times}. Furthermore, typically, at each triggering time, only the keys that appeared since the last triggering time are processed. However, this is problematic for DP -- if we only output a key when it has appeared in the most recent event time window, then this potentially leaks information. Naively, to avoid this, one has to process all the keys at each triggering time, which is computationally prohibitive.
    
    \item \mypar{Large number of observed keys and events} A fundamental challenge is scalability. The system should be able to handle millions of updates per second from a data stream with billions of distinct keys. 
    
    \item \mypar{Effective user contribution bounding} In real applications, each person may contribute multiple records to the data stream at different times. Providing ``event-level DP'' (where a single action by a person/user is protected) does not provide sufficient privacy protection. In this work, we provide ``user-level DP'', where~\emph{all} the actions by a person/user is protected simultaneously. To provide user-level DP, one has to bound the contribution of each user, and that eventually introduces bias on the statistics that get computed. But contribution bounding controls the variance of the noise we must add to ensure DP. A natural (and unavoidable \cite{amin2019bounding,liu2022histogram, kamath2023bias}) challenge is to decide on the level of contribution bounding to balance this bias and variance.
    
    \item \mypar{Streaming release of statistics} One has to output statistics at every triggering time. If we treat each triggering time as an independent DP data release, then the privacy cost grows rapidly with the number of releases. Alternatively, to attain a fixed DP guarantee, the noise we add at each triggering time must grow polynomially with the number of triggering times. This is impractical when the number of triggering times is large. Thus the noise we add to ensure DP is not independent across triggering times. This helps in drastically reducing the total noise introduce for a fixed DP guarantee. 
\end{enumerate}

In our design of DP-SQLP we address these challenges, either by designing new algorithms, or by implementing existing specialized algorithms. This is our main contribution. To the best of our knowledge,~\emph{we provide the first at-scale differentially private stream aggregation processing system}.

\mypar{Motivation for DP-SQLP} Data streams appear commonly in settings like Web logs, spatio-temporal GPS traces~\cite{andres2013geo}, mobile App usages~\cite{latif2020introducing}, data generated by sensor networks and smart devices~\cite{liu2018epic}, live traffic in maps~\cite{wang2016differential}, cardiovascular monitoring~\cite{tan2021toward}, real-time content recommendation~\cite{phelan2009using}, and pandemic contact tracing~\cite{cho2020contact}. Almost all of these applications touch sensitive user data on a continual basis. \citet{Calandrino11-like} demonstrated that continuous statistic release about individuals can act as a strong attack vector for detecting the presence/absence of a single user (in the context of collaborative recommendation systems). Hence, it is imperative to a streaming system to have rigorous privacy protections. In this work we adhere to differential privacy. For more discussion on the type of streams we consider, see a detailed survey in~\cite{isah2019survey}.

\mypar{Our Contributions} As mentioned earlier, our main contribution is overcoming multiple challenges to build a distributed DP stream processing system that can handle large-scale industrial workloads. 

\begin{itemize}
    \item \mypar{Private key selection} A priori, the set of possible keys is unbounded, so our system must identify a set of relevant keys to track. To protect privacy, we cannot identify a particular $\recordkey$ based on the contributions of a single user.  The streaming setting adds two additional complications: (a) The existence of each $\recordkey$ is only known when it is observed in a data record, and (b) the privacy leakage due to continually releasing information about any particular key increases the DP cost due to composition~\cite{dwork2014algorithmic}. To address these challenges, we design a novel algorithm (Algorithm~\ref{Alg:key-selection}) that couples ``binary tree aggregation''~\cite{Dwork-continual,CSS11-continual,honaker2015efficient} (a standard tool for continual release of DP statistics that only accumulates privacy cost that is poly-logarithmic in the number of aggregate releases from the data stream) with a thresholding scheme that allows one to only operate on keys that appear in at least $\mu>0$ user records. To further minimize privacy leakage, we employ a variance reduced implementation of the binary tree aggregation protocol~\cite{honaker2015efficient}. 
    
    \item \mypar{Preemptive execution} Since we may track a large number of keys in production systems, it is not scalable to scan through all of the state keys each time the system is invoked. Thus we design a new algorithm (Algorithm~\ref{Alg:empty-key-prediction}) that only runs on the keys that have appeared between the current and previous triggering times.
    The idea is to \emph{predict} when a key will be released in advance, rather than checking at each triggering time whether it should be released now.
    That is, whenever we observe a given $\recordkey$, we simulate checking the release condition for the rest of triggering times assuming no further updates to $\recordkey$. In the future we only check for $\recordkey$ at any triggering time if either of the two conditions happen: (i) $\recordkey$ appears in a fresh microbatch in the data stream, or (ii) the earlier simulation predicted a release for that time. By doing so, we reduce the expensive I/O and memory cost, with little CPU overhead. This idea is motivated by the caching of pages in the operating systems literature. 
    
    \item \mypar{Empirical evaluation} We provide a thorough empirical evaluation of our system. We consider a few natural baselines that adopt one-shot DP algorithms to stream data processing (e.g., repeated differential privacy query). At $(\varepsilon = 6, \delta=10^{-9})$-DP, we observed up to 93.9\% error reduction, and the number of retained keys is increased by 65 times when comparing DP-SQLP with baselines. Through our scalability experiments, we show that DP-SQLP can handle billions of keys without incurring any significant performance hit. 
    
    \item \mypar{Industry application} We present two industry use cases of streaming differential privacy. The DP-SQLP is applied to Google Shopping that processes a very large scale data stream under production environment, which demonstrates the scalability of our approach. A differentially private product page-view count is generated by DP-SQLP that can be used to signal the product information update. In the second use case, the streaming private key selection is applied to Google Trends for analyzing popular search queries with differential privacy guarantee.
    
\end{itemize}

In the following, we formally define the problem, and delve deeper into relevant related works.

\subsection{Problem Statement}
Let $D$ be a data stream defined by an unbounded sequence of records, i.e., $D=[d_1, d_2,...)$, where each record $d_i$ is a tuple 
$(\recordkey,$ $\recordvalue,$ $\recordtimestamp t_i,$ $\recorduserid)$. A common query pattern in data analytics is the unknown domain histogram query. For example, consider a web service that logs user activities: each record is a URL click that contains $(\texttt{URL},$ $\recorduserid,$ $\recordtimestamp)$. An analyst wants to know the number of clicks on each URL for each day up to today. An SQL query to generate this histogram is presented in Listing~\ref{target-sql1}.

\begin{lstlisting}[
           language=SQL,
           showspaces=false,
           basicstyle=\ttfamily,
           numberstyle=\tiny,
           commentstyle=\color{gray},
           caption=Single histogram query,
           captionpos=b,
           label=target-sql1,
           frame=single
        ]
SELECT URL, 
       TO_DATE(timestamp) AS date, 
       COUNT(*) AS count
FROM web_logs
GROUP BY URL, TO_DATE(timestamp)
\end{lstlisting}

In Listing \ref{target-sql1}, the keys are $(\texttt{URL}, \texttt{date})$ denoted by $\recordkey$\footnote{Throughout the paper, $\recordkey$ and k are used interchangeably.}, and the $\texttt{count}$ is an aggregation column denoted by $m$. 

\begin{figure}
\centering
\includegraphics[width=0.3\textwidth]{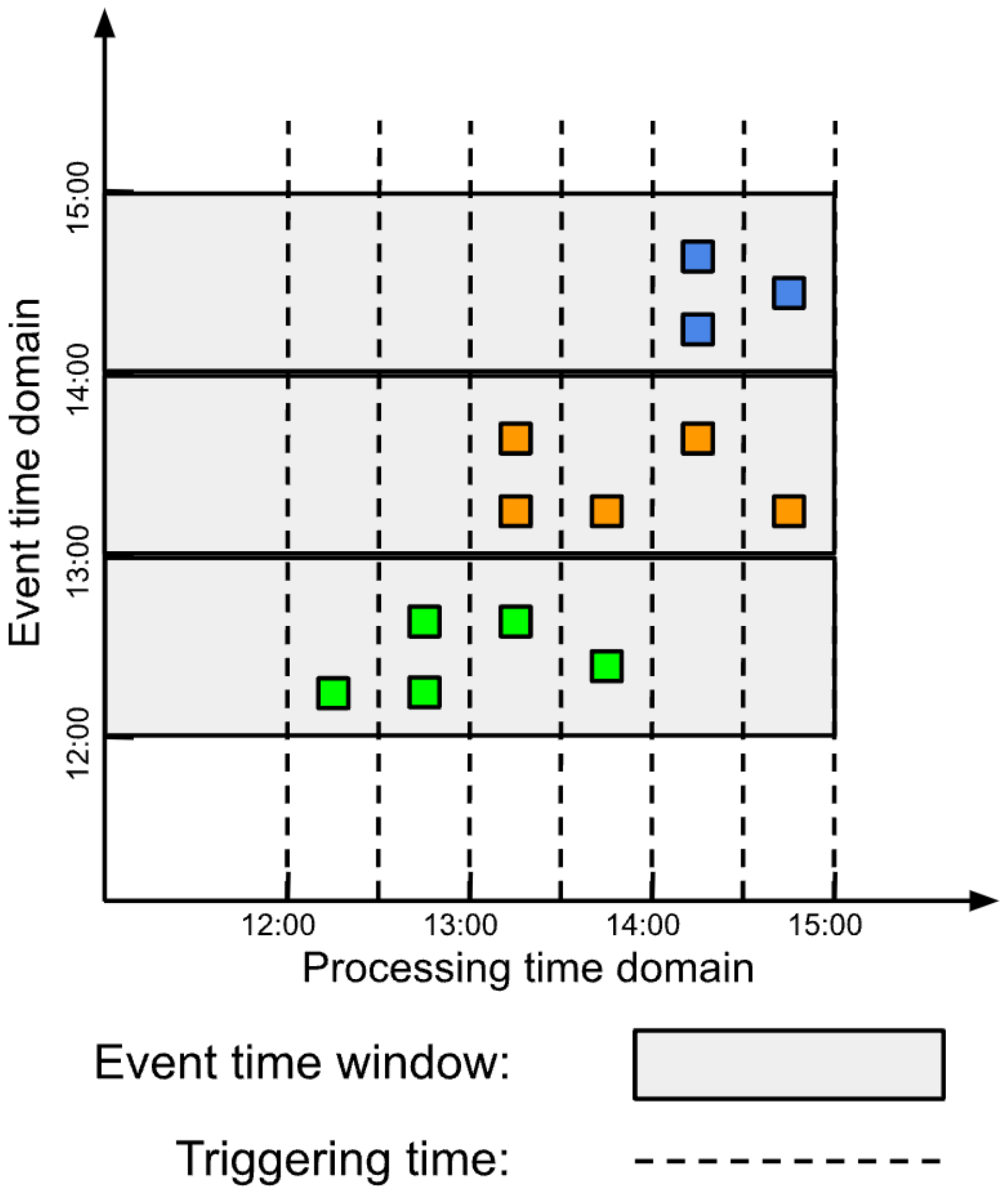}
\caption{Event time domain and processing time domain}
\label{fig:time-domain}
\end{figure}

When querying a growing database or data stream, the above-mentioned query only shows a snapshot at certain date. In stream data processing, we use event-time window to determine which records in the ~\emph{event time domain} to process, and triggers to define when in the ~\emph{processing time domain} the results of groupings are emitted \cite{akidau2015dataflow}.
Let $W$ denote the event-time windows, $W=[w_1, w_2...)$. Each event-time window is defined by a starting time and an end time $w_{i}=(t_{i,s}, t_{i,e})$. $D_{w_i} \in D$ contains all records that can be assigned to window $w_i$ so that the timestamp $t$ of each record satisfies $t_{i,s} \leq t < t_{i,e}$. Let $Tr$ denote a set of triggering times in the processing domain, $\tr = [tr_1, tr_2...)$. We assume triggering time is predefined and independent to dataset\footnote{In practical application, the streaming system may choose trigger adaptively, with complicated implementation~\citet{akidau2015dataflow}}. The streaming system incrementally processes $D_{w_i}$ at triggering time $\tr_{w_i} = \{tr_{i, s}, ..., tr_{i, e}\} \subset \tr$. Due to the time domain skew\cite{akidau2015dataflow}, $t_{i,s} \leq tr_{i,s}$ and $t_{i,e} \leq tr_{i,e} \leq t_{i, e} + t$, where $t$ is the maximum delay that system allows for late arriving records. Our goal is to release the histogram for all sub-stream $D_{w_{i}}$, at every triggering time $tr_i \in \tr_{w_{i}}$, in a differentially private (DP) manner (See Appendix~\ref{sec:dp-on-stream-appendix} for a formal DP definition). For a pictorial representations of various timing concepts, see Figure~\ref{fig:time-domain}.

\mypar{Privacy implication of input driven stream} In terms of privacy, we want to ensure that the stream processing system ensures $(\epsilon,\delta)$-user level DP~\cite{DMNS,ODO,dwork2014algorithmic} over the complete stream. Since we are operating under the constraint that the data stream is an input driven stream (Definition~\ref{def:input-driven}), the timings of the system (e.g., event time (Definition~\ref{def:event-time}), processing time (Definition~\ref{def:processing-time}), and triggering time (Definition~\ref{def:triggering-time})) can only be defined w.r.t. times at which the inputs have appeared. Thus it forces us to define the DP semantics which considers the~\emph{triggering times to be fixed} across neighboring data sets (in the context of traditional DP semantics). For a given user, what we protect via DP is the actual data that is contributed to the data stream. We provide a formalism in Appendix~\ref{sec:dp-on-stream-appendix}.

\begin{figure*}[ht]
\centering
\includegraphics[width=0.7\textwidth]{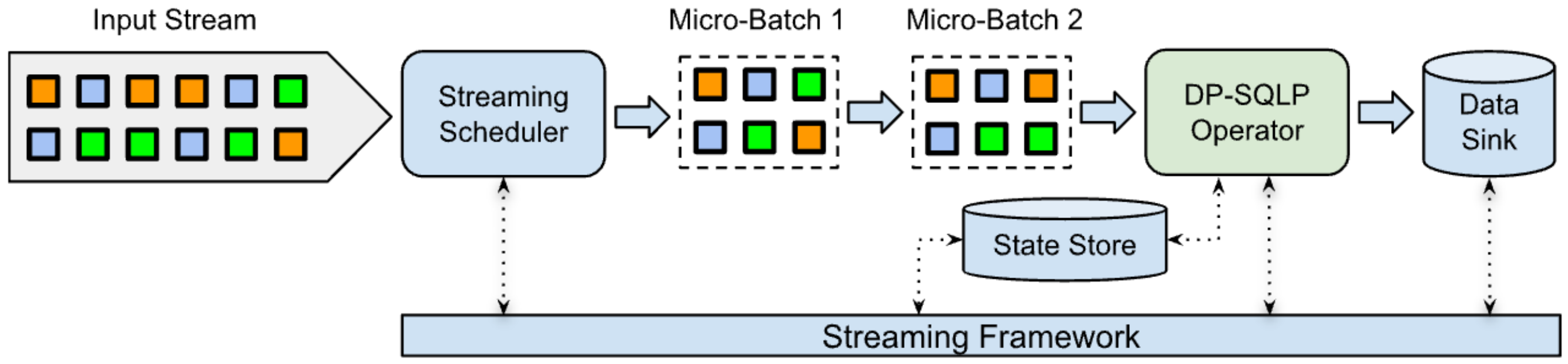}
\caption{High-level overview for the DP-SQLP system}
\label{fig:streaming-system}
\end{figure*}

\subsection{Related Work}
\label{sec:related}
Stream processing has been an active research field for more than 20 years~\cite{fragkoulis2020survey}. It is now considered to be a mature technology with various streaming frameworks deployed at scale in industry, including Spark Streaming~\cite{Zaharia13streamingspark}, Apache Beam~\cite{akidau2015dataflow} and Apache Flink~\cite{DBLP:CarboneKEMHT15flink}. However, none of these systems offer differentially private streaming queries.

Our work builds on a long line of DP research that focuses on extending one shot applications of DP mechanisms to the continual observation (streaming) setting for both analytics and learning applications \cite{Dwork-continual,CSS11-continual,honaker2015efficient, cardoso2022differentially, kairouz2021practical, li2015matrix, denisov2022improved, henzinger2022constant, epasto2023differentially}. These mechanisms crucially leverage the tree aggregation protocol \cite{Dwork-continual,CSS11-continual,honaker2015efficient} or variants of it based on the matrix factorization mechanism \cite{li2015matrix, denisov2022improved, henzinger2022constant}. All these approaches have the advantage of drastically reducing the error induced by repeated application of a DP mechanism, making the DP protocol itself stateful. 

Of the above cited, our work is most related to \cite{cardoso2022differentially}, which investigates the problem of computing DP histograms with unknown-domains under continual observations. They leverage (extensions of) the tree aggregation protocol to build an efficient DP protocol for the continual observation setting. Our key selection algorithm (Algorithm~\ref{Alg:key-selection}) is heavily inspired by~\cite{cardoso2022differentially}, with the main difference being the use of additional threshold $\mu$ for further privacy protection, and an algorithm to predict whether to select a key even if it has zero records (Algorithm~\ref{Alg:empty-key-prediction}). Algorithm~\ref{Alg:empty-key-prediction} is crucial to the scalability of our approach to production workloads. From a system design point, our work further extend~\cite{cardoso2022differentially} as follows:
\begin{enumerate}
    \item We consider user level DP. As outlined in the introduction, this introduces interesting algorithmic and system challenges.
    \item We provide a concrete streaming system architecture for scalable production deployments whereas they focus on developing algorithms. More precisely, we develop and test an empty key prediction scheme that allows us to scale to millions of updates per second that contain billions of distinct keys.
    \item We test our algorithms and architecture on a number of large-scale synthetic and real-life datasets to demonstrate the efficacy of our approach relative to meaningful baselines.
    \item Our system is deployed to a extreme large scale production use case.
\end{enumerate}

Another prior system that continuously releases streaming user count at scale is Google FLEDGE k-Anonymity server \cite{FLEDGE,epasto2023differentially}. It determines whether a given advertisement has been shown to at least k users over a sliding window with event-level differential privacy guarantee. While our algorithms and systems are more general that can be applied to arbitrary histogram query at scale, with a user-level guarantee. An extended list of real-world uses of DP is presented in~\citet{DamienBlog}.  

\subsection{Organization}

The rest of the paper is organized as follows. In Section~\ref{sec:preliminaries} we provide the necessary background on differential privacy and streaming systems; in Section~\ref{sec:sdp-overview} we describe the main algorithmic components of our DP streaming system; in Section~\ref{sec:streamingdpimp} we provide details of the improvements needed to scale up the algorithms to large workloads; in Section~\ref{sec:experiments} we provide a thorough experimental evaluation; and finally in Section~\ref{sec:conclusion} we provide some concluding remarks and outline a few interesting open directions. We also provide a glossary of terms from the streaming literature (used in this paper) in Appendix~\ref{sec:glossary}.

\section{Preliminaries}
\label{sec:preliminaries}

In this section, we describe the formalism that will be necessary for the rest of the paper: a) DP on Streams (Appendix~\ref{sec:dp-on-stream-appendix}), b) DP continual observation (Appendix~\ref{sec:dp-tree-appendix}), and c) System architecture for the streaming system (Section~\ref{sec:sysArch}). For the purpose of brevity, we will defer (a) and (b) to the appendix. For a comprehensive introduction to DP, please refer to~\cite{vadhan2017complexity}.

% \subsection{Differential Privacy on Streams}
% \label{sec:dpdef}

% Refer to Appendix~\ref{sec:dp-on-stream-appendix}.

% \subsection{Privacy Under Continual Observation and Binary Tree Aggregation} 
% \label{sec:backTrre}

% Refer to Appendix~\ref{sec:dp-tree-appendix}.

\subsection{Streaming System Architecture}
\label{sec:sysArch}

\begin{sloppypar}
The streaming differential privacy mechanism described in this paper can be generally applied to various streaming frameworks, including Spark Streaming\cite{Zaharia13streamingspark}, Apache Beam\cite{akidau2015dataflow} and Apache Flink\cite{DBLP:CarboneKEMHT15flink}. The DP-SQLP system we develop is implemented using a streaming framework similar to Spark Streaming\cite{Zaharia13streamingspark}, as shown in Figure~\ref{fig:streaming-system}.
\end{sloppypar}

The input data stream contains unordered, unbounded event records. The streaming scheduler will first assign each record to the corresponding event-time window $w$. Within each window, at every triggering timestamp $tr$, records are bundled together to create an immutable, partitioned datasets, called a~\emph{micro-batch}. After a micro-batch is created, it will be dispatched to the DP-SQLP operator for processing.

When processing a micro-batch, the DP-SQLP operator will interact with the system state store for a state update. Once the differentially private histogram is generated, it will be materialized to the data sink\footnote{Data sink is the storage system used to store and serve output data, including file systems and databases.}. Similar to Spark Streaming, our streaming framework provides consistent, "exactly-once" processing semantic across multiple data centers. In addition, the streaming framework also provides fault tolerance and recovery.

There are multiple ways to schedule micro-batches based on certain rules\cite{akidau2015dataflow}, like processing timer based trigger, data arrival based trigger and combinations of multiple rules. 

Based on the number of micro-batches received by operator at each time instance, we can also classify scheduling methods into two categories, sequential scheduling and parallel scheduling, as shown in the Figure~\ref{fig:streaming-schedule}. In sequential scheduling, input data stream is divided into a sequence of micro-batches. The operator will process one micro-batch at a time. Parallel scheduling is able to further scale up the pipeline, by allowing multiple micro-batches to be processed at the same time. The streaming scheduler will partition records by predefined key tuples, and create one micro-batch per key range.

\begin{figure}
\centering
\includegraphics[width=0.3\textwidth]{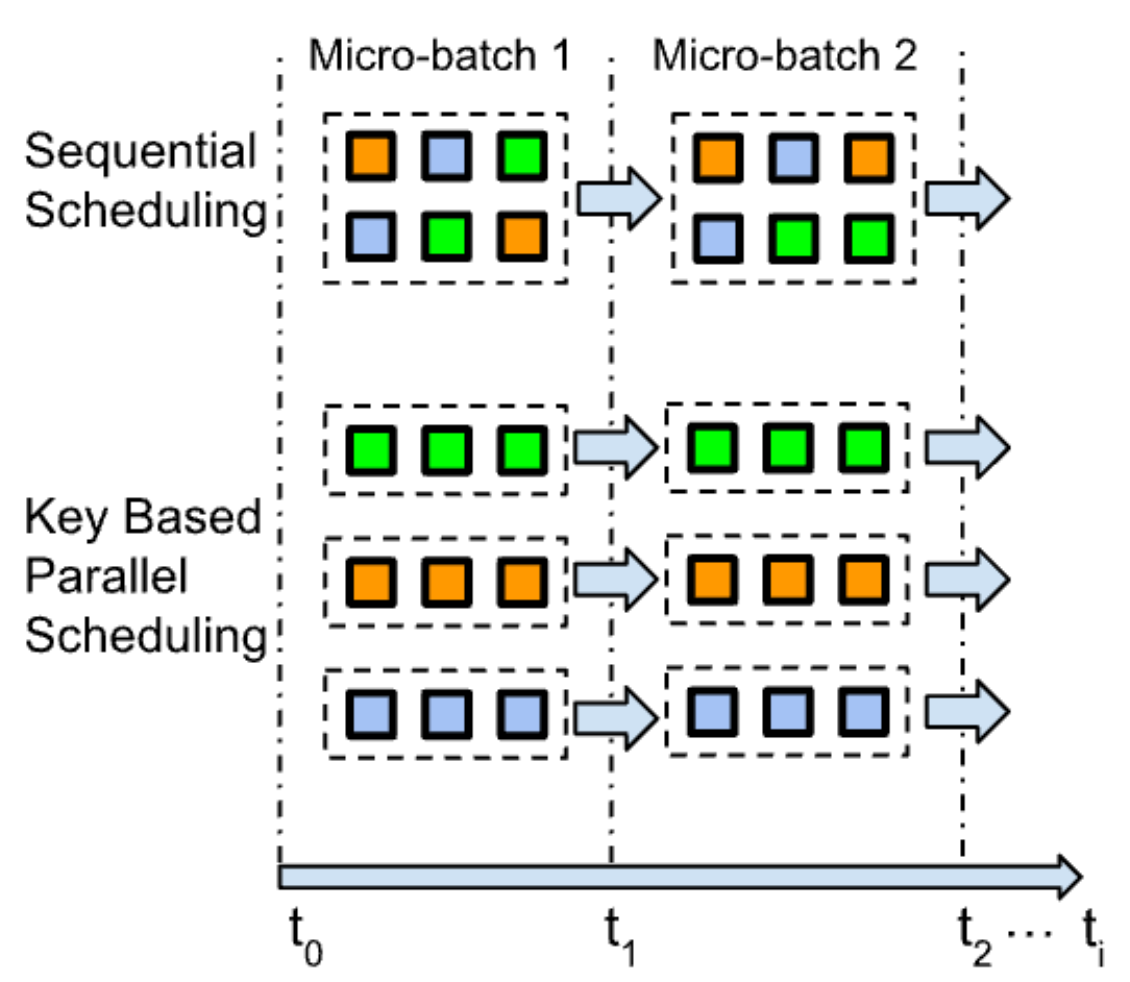}
\caption{Two types of streaming scheduling}
\label{fig:streaming-schedule}
\end{figure}

In the rest of paper, we will assume the sequential scheduling for the algorithm discussion. Given data stream $D_{w_i}$ for window $w_i$, the sub-stream at triggering timestamp $tr_i \in \tr_{w_i}$ can be denoted as $D_{tr_i} \subseteq D_{w_i}$, and the sub-stream for each micro-batch can be represented as the incremental data stream $\Delta D_{tr_i} = D_{tr_i} - D_{tr_{i-1}}$.

As Akidau pointed out in the unified dataflow model \cite{akidau2015dataflow}, batch, micro-batch, and pure streaming are implementation details of the underlying execution engine. Although the streaming differential privacy mechanism discussed in this paper is executed by a streaming system based on micro-batch, the mechanism and algorithm can be widely applied to batch, micro-batch, and pure streaming systems.

It is worth noting that different execution modes (batch, micro-batch and pure streaming) result in different trade-offs between data utility and pipeline latency. In differential privacy, the more frequent we repeat the process, the nosier results tend to be. Therefore, data utility is an additional factor when choosing execution modes.

\begin{figure*}[ht]
\centering
\includegraphics[width=0.8\textwidth]{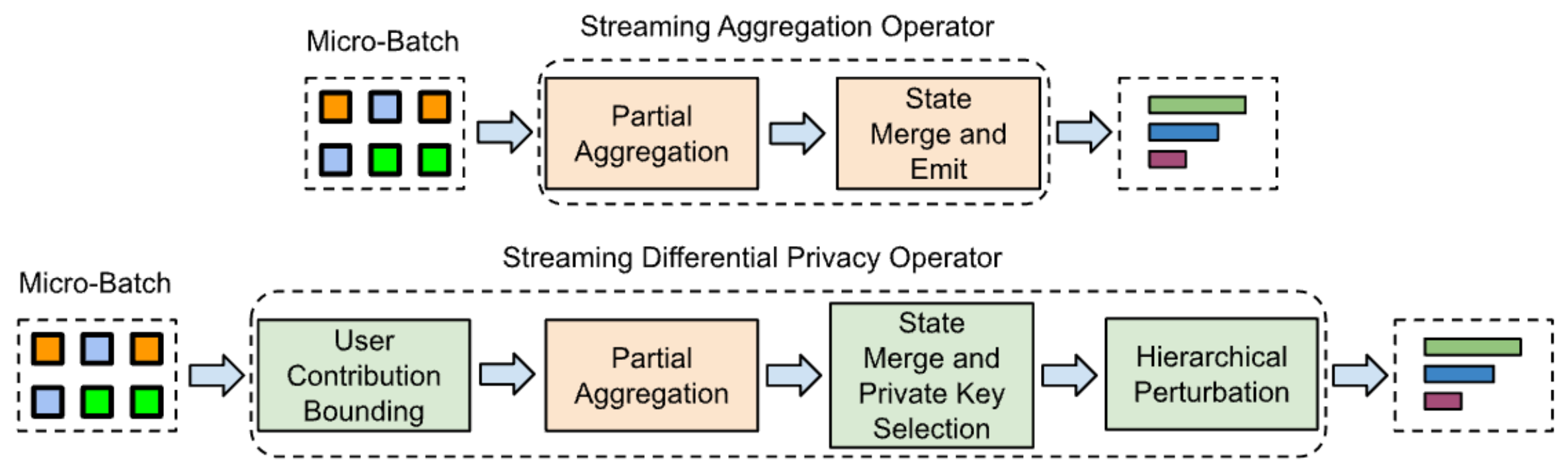}
\caption{Overview of streaming differential privacy mechanism}
\label{fig:sdp-overview}
\end{figure*}

\section{Streaming Private Mechanism}
\label{sec:sdp-overview}

In this section we will discuss the overall mechanism for streaming differential privacy. Our target is to perform aggregation and release histogram at every triggering timestamp in $\tr$, while maintaining $(\epsilon, \delta)$-\textit{differential privacy}. There are four main components within streaming differential privacy mechanism -- user contribution bounding, partial aggregation, streaming private key selection and hierarchical perturbation, as shown in Figure~\ref{fig:sdp-overview}. 

To simplify the discussion, we will assume $\Sum$ is the aggregation function for $\groupby$. It is also possible to use other aggregation function within streaming differential privacy mechanism.

Let's declare the inputs, parameters and outputs that will be used in streaming differential privacy.

\begin{itemize}
    \item Input: Data Stream $D$, event-time windows $W$, triggering timestamp per window $\tr_w$, privacy parameters $\epsilon, \delta > 0$, accuracy parameter $\beta > 0$, per record clamping limit $L$.
    \item System Parameters: Max. no. of records per user $C$.
    \item Output: Aggregated DP histogram at every triggering timestamp.
\end{itemize}

\subsection{Non-Private Streaming Aggregation}
The traditional streaming aggregation operator without differential privacy is shown on the top of Figure \ref{fig:sdp-overview}. Records within the micro-batch are grouped by key and aggregated by the \textit{reduce} function \cite{Zaharia13streamingspark}. After that, the partial aggregation result will be merged with the previous state \cite{to2018survey} and the update histogram is emitted. There is no differential privacy protection within this process, and user privacy can be leaked from multiple dimensions, including histogram value, aggregation key and the differences between two histogram updates.

\subsection{User Contribution Bounding}
\label{sec:user:contribbound}

DP algorithms require that the sensitivity of each user contributions be limited, for example that a user can contribute up to $C$ times. However, in reality, each user may contribute to many records and many keys, especially for heavy users. Therefore, we need to bound the maximum influence any user can have on the output in order to achieve a desired overall DP guarantee. This step is called ``user contribution bounding''.

Some one-shot mechanisms perform user contribution bounding by limiting contributed value per key and the number of contributed keys per user~\cite{amin2022plume,wilson2019differentially}. However, this approach does not fit the streaming setting, since it requires three shuffle stages - shuffle by user, shuffle by (key, user) and shuffle by key. 

User contribution bounding in streaming DP is performed on the user level for the entire data stream $D$\footnote{In a production streaming system, the DP guarantee is commonly defined with the minimum privacy protection unit (e.g., [user, day]). The maximum number of record per user $C$ needs to be enforced within each privacy unit.}:

\begin{itemize}
    \item Each user can contribute to at most $C$ records in the data stream $D$.
    \item The value $v$ for the aggregation column $m$ in each record is clamped to $L_m$ so that $|v| < L_m$.
\end{itemize}

The maximum number of records per user $C$ and per record clamping limit $L_m$ together determine the per-user $\ell_1$ sensitivity in data stream:

$$L_1 = C \times L_m.$$

Choosing the right contribution bounding parameters $C$ and $L_m$ is critical for privacy-utility trade-off. When the bounding limit is small, the noise is small, but the data loss may be significant (e.g. if most/all users have a lot of data to contribute). On the other side, when the bounding limit is large, the data loss is small, but noise is large. It is possible to find a near optimal point from a heuristic study, which we discuss in Section \ref{sec:experiments}. Indeed, one approach to choosing a good contribution bound is to inspect the data distribution. For example, we can pick $C$ at $99^{th}$ percentile of per-user records (i.e. $<1\%$ users have more than C records), which can be chosen in a DP way if computed on a fraction of the data stream, or in a non-DP way if it is based on proxy data.

In the following sections, we will introduce streaming private key selection and hierarchical perturbation, which are two main private operations in streaming differential privacy. Since the private operations are performed per window, we will focus on $D_{w_i}$ during algorithm discussion. However, the same private operations should be applied to all windows.

\subsection{Streaming Private Key Selection}
\label{sec:partkeyselection}

The main objective we consider in this section is to select the set of keys that exceed a certain threshold of user contributions $\mu\geq 0$. (These are the keys that are used later for releasing the aggregation columns in Section~\ref{sec:hierarchical} below.) Recall that our streaming system is input driven by a growing data stream, meaning, one only sees the existence of a key if it has at least one user contribution. This poses a significant privacy challenge since the detectable set of keys is highly dependent on the data set. As a result, we design a novel thresholding scheme coupled with binary tree aggregation~\cite{Dwork-continual,CSS11-continual,honaker2015efficient} that allows one to only operate with keys that deterministically have at least $\mu$ user contributions, and still preserve $(\epsilon,\delta)$-DP. In the following, we provide a description of the algorithm, along with the privacy analysis.

\mypar{Data preprocessing} After user contribution bounding (discussed in Section~\ref{sec:user:contribbound}), we first perform a regular key \groupby and aggregation for all records within the current micro-batch\footnote{This is a system level operation without any implication to the privacy guarantee.}. Then we merge the aggregated data on each key with a data buffer that is stored in a system state\footnote{For every key, we accumulate aggregation column values, as well as the number of unique users. This process is called data accumulation.}. Beyond that we execute the key selection algorithm described in Algorithm~\ref{Alg:key-selection}.

\mypar{Algorithm description} As mentioned earlier, the emitted key space is not predefined. A key may emerge when there is at least one user record with the key (due to the nature of input driven stream). Therefore, the streaming DP system must determine \emph{when} and \emph{what} keys to release or update, in a private manner. This is different from the non-private streaming aggregation, where updates will be emitted at every triggering timestamp after processing each micro-batch.

\mypar{Remark} In the description of Algorithm~\ref{Alg:key-selection}, we use the following primitives implicitly used in Algorithm~\ref{Alg:abcd2}: i) \\$\init(T,\sigma)$: Initialize a complete binary tree $\tree$ with $2^{\lceil T\rceil}$ leaf nodes with each node being sampled from $\calN(0,\sigma^2)$, ii) $\addt(\tree,i,c_i)$: Add $c_i$ to all the nodes on the path from the $i$-th leaf node to the root of the tree, and iii) $\gettotal(\tree_\recordkey,i)$: Prefix sum of the all the inputs $\{c_1,\ldots,c_i\}$ to the binary tree computed via Algorithm~\ref{Alg:abcd2}.

Our approach is an extension of thresholding algorithm in \cite{korolova2009releasing} to the streaming setting. In short, we carefully select a threshold and compute (with DP noise) the number of unique users that contribute to each encountered key. If this noisy number is greater than or equal to the chosen threshold, the key is released. We describe the algorithm in full detail in Algorithm~\ref{Alg:key-selection}, and provide the formal privacy guarantee in Theorem~\ref{thm:privKeySelection}.

A crucial component of the algorithm is the choice of the threshold $\tau$ in Line~\ref{line:thresh} of Algorithm~\ref{Alg:key-selection}. One can instantiate $\tau$ with the bound in Theorem~\ref{thm:utilGuaranteeTree}. However, in our implementation (described in Section~\ref{sec:streamingdpimp}) we actually implement the tree aggregation via the ``bottom-up Honaker'' variance reduction described in Appendix~\ref{sec:dp-tree-appendix}. One can write the exact distribution of the differences between DP-Tree aggregated count and the true count, which is $\hat{q}_{{tr_{i},k}}-\countt_k(D_{tr_{i}})$ in  Line~\ref{line:thresh} of Algorithm~\ref{Alg:key-selection}, via~\eqref{eq:distTree}. This allows us to get a tighter bound on $\tau$ based on the inverse CDF of the Gaussian distribution. Also, it should be obvious from \eqref{eq:distTree} that the variance of the Gaussian distribution is dependent on the time step at which we are evaluating the cumulative sum. Hence, to obtain a tighter estimation of the threshold, we actually have a time dependent threshold $\tau_{tr_{i}}$ (based on~\eqref{eq:distTree}) instead of an universal threshold in Line~\ref{line:thresh}.

\mypar{Remark} For brevity, in Theorems~\ref{thm:privKeySelection} and~\ref{thm:utilTreeKey}, we provide the guarantees assuming each user only contributes once, i.e., in the language of Section~\ref{sec:user:contribbound} we assume that that user contribution bound $C=1$. However, in our implementation we do allow $C>1$. The idea is to use a tighter variant of advanced composition for $(\epsilon,\delta)$-DP~\cite{dwork2014algorithmic} while ensuring that each user contributes~\emph{at most once to each key} in any instance of Algorithm~\ref{Alg:key-selection}. In the following we provide the privacy and the utility guarantees. For brevity, we defer the proofs to Appendix~\ref{sec:theorem31-proof-appendix} and Appendix~\ref{sec:theorem32-proof-appendix}.

\begin{algorithm}[ht]
\caption{Streaming Private Key Selection}
\begin{algorithmic}[1]
\REQUIRE Data stream $D_{w_i}=\{d_1,\ldots,d_n\}$, where the event timestamp $t_i$ of each $d_i$ can map to event-time window $w_i=(t_s, t_e)$, $t_s < t_i < t_e$, triggering timestamps $\tr_{w_i} = [tr_{1}, tr_{2},...,tr_{T_i}]$. At each triggering time $tr_{i} \subseteq \mathbb{R}$, only a sub-stream $D_{tr_{i}} \subseteq D_{w_i}$ is available. Threshold $\thresh \ge 0$, privacy parameters $\epsilon,\delta > 0$, failure probability $\beta >0$.
{\STATE Compute the noise standard deviation $\sigma$ for the tree aggregation based on $(T_i=|\tr_{w_i}|, \epsilon, \delta)$. (See Appendix~\ref{app:key-selection} for more details.)\label{line:keySelection1}}
{\STATE Compute the accuracy threshold $\tau$ of the tree aggregation such that for any fixed $\recordkey$ and the corresponding binary tree $\tree_\recordkey$, $$\pr{\tree_\recordkey}{\forall tr_{i}\in\tr_w, |\hat{q}_{{tr_{i},\recordkey}}-\countt_\recordkey(D_{tr_{i}})| \le \tau} \ge 1-\beta,$$ which depends on $\sigma$ and $\beta$. (See Appendix~\ref{app:key-selection} for more details.) Here, $\countt_\recordkey(D_{tr_{i}})$ denote the unique user count for $\recordkey$ in $D_{tr_{i}}$, and $\hat{q}_{{tr_{i},\recordkey}}$ denote the private estimate of $\countt_\recordkey(D_{tr_{i}})$.\label{line:thresh}}
\FOR{$i\in|\tr_{w_i}|$}
    \STATE $\calS^{(i)}\leftarrow$ Set of all keys in the stream $D_{tr_{i}}$ with count $>\mu$.
    \STATE For all $\recordkey\in\calS^{(i)}\backslash\calS^{(i-1)}$, create a new tree $\tree_{\recordkey}$ using $\init(T_i,\sigma)$, and execute Algorithm~\ref{Alg:abcd2} till $(i-1)$-th step with all zeros as input. 
    \FOR{$\recordkey\in\calS^{(i)}$}
        \STATE $\tree_\recordkey \leftarrow \addt (\tree_\recordkey,i, \countt_\recordkey(D_{tr_{i}}) -$ $\countt_\recordkey(D_{tr_{i-1}}))$, i.e., Add the count at time stamp $tr_{i}$  to $\tree_{\recordkey}$.
        \STATE $\hat{q}_{{tr_{i},\recordkey}} \gets \gettotal(\tree_\recordkey,i)$.
        {\STATE {\bf if} $\hat{q}_{{tr_{i},\recordkey}} > \mu + \tau$, {\bf then } output $(\recordkey,\hat{q}_{{tr_{i},\recordkey}})$.\label{line:selection-continue}}
    \ENDFOR
\ENDFOR
\end{algorithmic}
\label{Alg:key-selection}
\end{algorithm}

\begin{thm}[Privacy guarantee]
    Algorithm \ref{Alg:key-selection} is $(\varepsilon, \delta + (e^\varepsilon+1)\cdot\beta)$-DP for addition or removal of one element of the dataset.
    \label{thm:privKeySelection}
\end{thm}

\begin{thm}[Utility guarantee]
For any fixed key $k$, there exists a threshold $\tau=\bigOh{\frac{\sqrt{\log(T/\beta)\log^2(T)\log(1/\delta)}}{\epsilon}}$ such that w.p. at least $1-\beta$, Algorithm~\ref{Alg:key-selection} outputs $k$ if at any one of the triggering time (in $\tr = [tr_{1}, tr_{2},...,tr_{T}]$) the true count is at least $\mu+\tau$ .
\label{thm:utilTreeKey}
\end{thm}

\subsection{Hierarchical Perturbation}
\label{sec:hierarchical}

\begin{algorithm}[ht]
\caption{Hierarchical Perturbation with DP-Tree}
\begin{algorithmic}[1]
\REQUIRE Data stream: $D_{w_i}=\{d_1,\ldots,d_n\}$, where each $d_i$ arrive at time $t_i$ within event-time window $w_i=(t_s, t_e)$, $t_s < t_i < t_e$. Triggering timestamps $\tr_{w_i} = [tr_{1}, tr_{2},...,tr_{T_i}]$. At each triggering timestamp $tr_{i} \subseteq \mathbb{R}$, only a sub-stream $D_t \in D_{w_i}$ is available. Privacy parameters $\epsilon,\delta > 0$, number of $\DpTree$ leaf nodes $n$.
\STATE Compute the noise standard deviation $\sigma$ for the tree aggregation based on $(n, \epsilon, \delta)$. (See Appendix~\ref{app:key-selection} for more details.)
\FOR{$i\in|\tr_{w_i}|$}
    \STATE $\calS_i\leftarrow$ Set of keys output by the key selection algorithm (Algorithm~\ref{Alg:key-selection}) at $tr_{i}$. 
    \FOR{$\recordkey\in\calS_i$}
        \STATE Let $D_{tr_{i}}\leftarrow$ data stream available at time stamp $tr_{i}$.
        \STATE Last Release Time $\lrt_\recordkey\leftarrow$ triggering timestamp of the previous statistic release.
        \STATE $\Delta V_\recordkey\leftarrow$ Aggregated value for $\recordkey$ in the sub-stream $D_{tr_{i}}-D_{\lrt_\recordkey}$.
        \STATE $\tree_\recordkey\leftarrow\addt(\tree_\recordkey,i,\Delta V_\recordkey)$.
        \STATE  Output $\gettotal(\tree_\recordkey,i)$.
    \ENDFOR
\ENDFOR
\end{algorithmic}
\label{Alg:dptree-aggregation}
\end{algorithm}

Once sufficient records of certain key are accumulated (i.e., following the notation from the previous section there are at least $\mu$ unique user contributions), the objective is to select the key for \emph{statistic release}. Statistic release corresponds to adding the value from user contributions to a main histogram that estimates the distribution of records across all the keys. Notice that in this histogram, a single user can contribute multiple times to the same key. In this section we discuss how to create this histogram while preserving DP. 

The crux of the algorithm is that for all the set of keys detected during the key selection phase via Algorithm~\ref{Alg:key-selection}, we maintain a DP-tree (an instantiation of Algorithm~\ref{Alg:abcd2}) for every key detected. We provide a $\rho$-zCDP guarantee for each of the trees for each of the keys. Since each user can contribute $C$ records in this phase, and in the worst case all these contributions can go to the same node of a single binary tree, we scale up the sensitivity corresponding to any single node in the tree to $L_1=C\cdot L$ (analogous to that in Section~\ref{sec:user:contribbound}), and ensure that each tree still ensures $\rho$-zCDP. We provide the details of the algorithm in Algorithm~\ref{Alg:dptree-aggregation}. The privacy guarantee follows immediately from Theorem~\ref{thm:privGuaranteeTree}, and the translation from $\rho$-zCDP to $(\epsilon,\delta)$-DP guarantee. In Algorithm~\ref{Alg:dptree-aggregation}, we will use a lot of the binary tree aggregation primitives we used in Section~\ref{sec:partkeyselection}.

\section{System Implementation and Optimization}
\label{sec:streamingdpimp}

\begin{figure*}[ht]
\centering
\includegraphics[width=\textwidth]{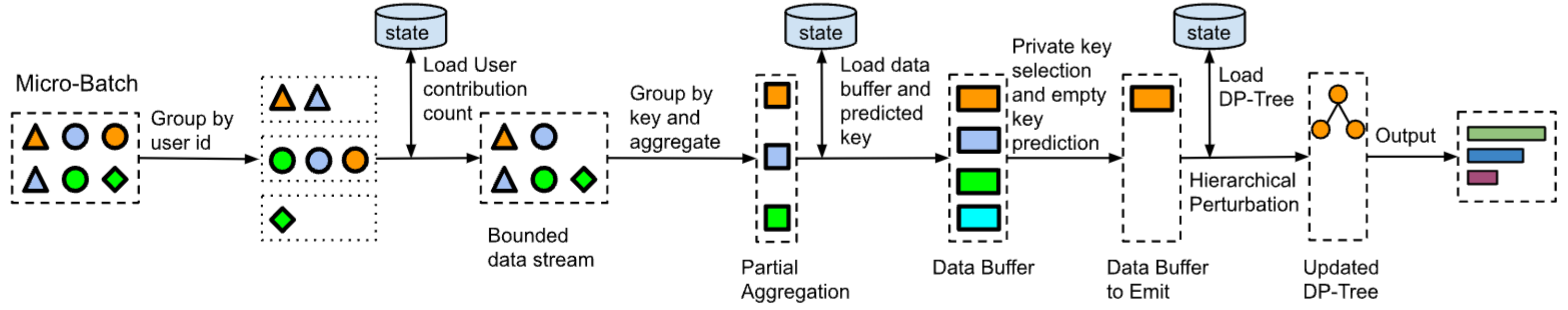}
\caption{Execution of streaming differential privacy mechanism}
\label{fig:sdp-implementation}
\end{figure*}

When implementing the streaming differential privacy algorithms described in Section \ref{sec:sdp-overview}, one must take the system constraints into practical considerations. There are three main challenges:

\begin{itemize}
    \item The streaming framework described in Section~\ref{sec:sysArch} discretizes the data stream into micro-batches. Therefore, streaming key selection and hierarchical perturbation, whose algorithms are defined based on data stream $D_{t^{\tr}_i}$, must be implemented using micro-batch $\Delta D_{t^{\tr}_i}$ (defined in Section~\ref{sec:sysArch}) and the system state. In Section~\ref{sec:statemanagement} we detail the complete state management of DP-SQLP.
    \item DP-SQLP is input data stream driven. The state loading and updating require the existence of a key in the current micro-batch. However, Algorithm ~\ref{Alg:key-selection} requires to test all keys that have appeared at least once. In Section~\ref{sec:emptyKey} we discuss a new algorithm that avoids testing all the keys that have appeared at least once.
    \item There are multiple components to the DP-SQLP system which are individually $(\epsilon,\delta)$-DP. It is necessary to use appropriate forms of composition to account for the total privacy cost. In Section~\ref{sec:accounting}, we detail the complete privacy accounting for DP-SQLP.
\end{itemize}

\subsection{State Management}
\label{sec:statemanagement}

As mentioned above, both streaming key selection and hierarchical perturbation are defined based on data stream $D_{t^{\tr}_i}$. Therefore, they both require stateful operations. Furthermore, the global user contribution bounding that tracks the number of records per-user in data stream $D$ is also stateful.

In DP-SQLP, the system state store is a persistent storage system, backed by Spanner database~\cite{corbett2013spanner} that provides high availability and fault tolerance. The system state store is co-managed by the DP-SQLP operator and the streaming framework for state update and maintenance. All the state information required by streaming differential privacy is stored in the system state store. For a pictorial depiction, see Figure~\ref{fig:streaming-system}.

\mypar{State Store Structure} There are two main state tables in the system state stores, which are managed by the same streaming framework. Each state table is a key value storage containing ~\emph{state key} and ~\emph{state object}.

The first state table is keyed by user id to track per-user contribution within the data stream. The state object simply stores the count value.

The second state table is keyed by \groupby keys. The state object contains data buffer, \DpTrees for key selection and \DpTrees for aggregation columns. Data buffer is a data structure that temporarily stores the unreleased, aggregated data from new records, due to the failure in thresholding test (line~\ref{line:selection-continue}, Algorithm~\ref{Alg:key-selection}). One \DpTree is used by each round of Algorithm~\ref{Alg:key-selection} execution, and one \DpTree is used for hierarchical perturbation per aggregation column.

\mypar{Execution Procedures} The execution of streaming differential privacy mechanism is shown in Figure~\ref{fig:sdp-implementation}. Different shapes represent different users and different colors represent different keys. Each step is described as following.

\begin{enumerate}
    \item \textbf{User contribution bounding}: Records in one~\emph{micro-batch} are grouped by user id. A map in~\emph{system state store} is maintained to track the number of records each user contributes. Once the number of  contributions for a user reaches $C$, all the remaining records for that user in the data stream will be discarded. Furthermore, we clamp the value $v$ of each~\emph{aggregation column} $m$, so that $|v| \leq L_m$.
    \item \textbf{Cross-user aggregation}: Records in one~\emph{micro-batch} are grouped by key and aggregated, which form a delta result~\cite{fegaras2016incremental}. After that, the delta result will be merged into the data buffer that is loaded from the system state store. 
    \item \textbf{Streaming key selection}: The \DpTrees for streaming key selection are loaded from the system state store. Then we will perform Algorithm~\ref{Alg:key-selection}, which adds the incremental user count from the current micro-batch into \DpTree, as a leaf node.
    \item \textbf{Hierarchical perturbation}: Once a key is selected, the \DpTree for hierarchical perturbation is loaded from the system state store. After that, we will use Algorithm~\ref{Alg:dptree-aggregation} to get DP aggregation results, and output the results.
\end{enumerate}

The execution engine used to implement user contribution bounding and hierarchical perturbation will be discussed in section~\ref{sec:engine}.

As mentioned in section~\ref{sec:partkeyselection}, the \DpTree estimator is implemented with the “bottom-up Honaker” variance reduction to get the DP sum. The estimated sum for \DpTree root at ${\sf node}_{i^*}$ equals 
$$
\hat\Sum({\sf node}_{i}) = \sum\limits_{j=0}^{\mu-1} c_j\cdot {\sf sum}(\level_j).
$$
In case more than one \DpTrees are used in key selection or hierarchical perturbation, we need to further sum the Honaker estimations from each tree together.

\subsection{Parallel Execution}
\label{sec:engine}
\begin{figure}[ht]
\centering
\includegraphics[width=0.3\textwidth]{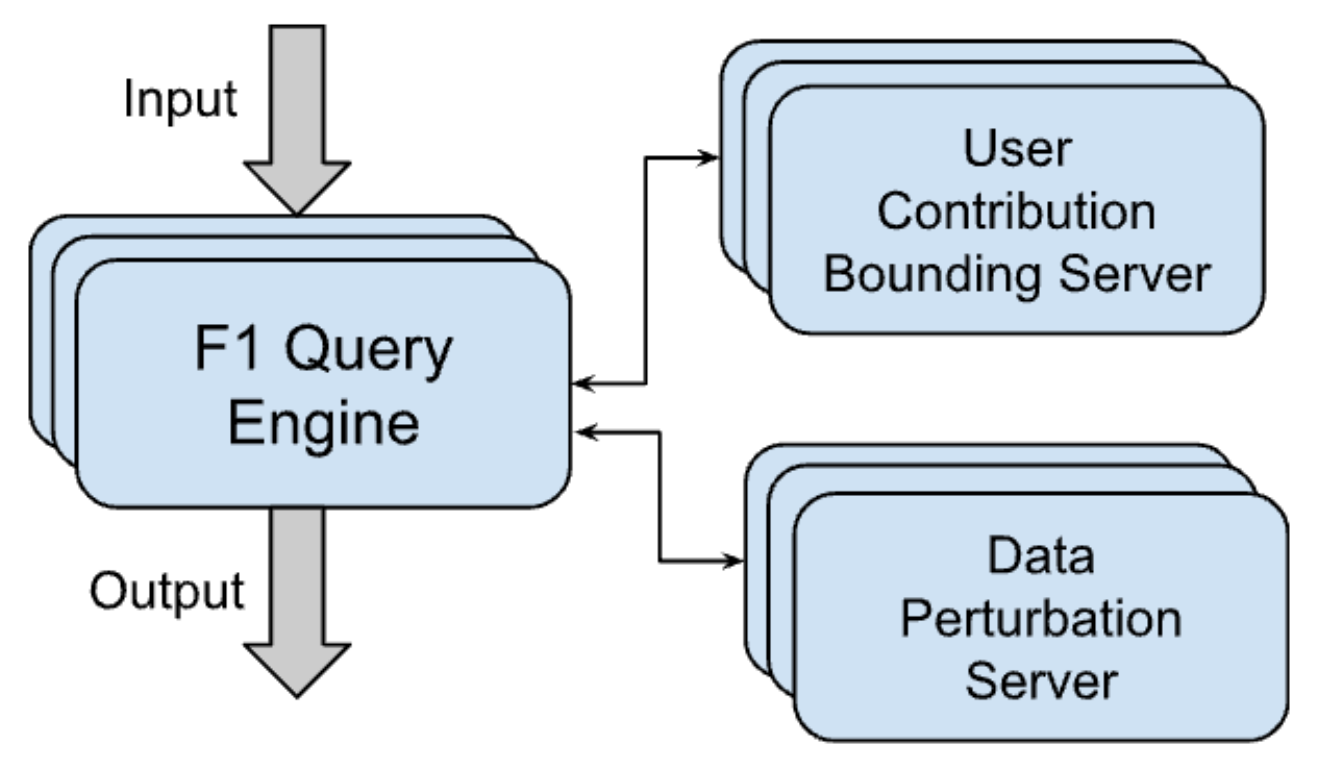}
\caption{Parallel execution within DP-SQLP operator}
\label{fig:execution}
\end{figure}

When building the DP-SQLP operator, we leverage the F1 query engine~\cite{samwel2018f1} for its wide range of data sources and distributed query execution (Figure~\ref{fig:execution}). The user contribution bounding step is executed by the~\emph{user contribution bounding server}. The privacy key selection and hierarchical perturbation are executed by the~\emph{data perturbation server}. Both servers contain thousands of workers that are horizontally scalable. Input data is first read by F1 query engine, partitioned, then sent to the user contribution bounding server through Remote Procedure Calls (RPCs). After that, the bounded data will stream back to F1, re-partitioned, then being sent to data perturbation server for key selection and hierarchical perturbation.

\subsection{Empty Key Release Prediction}
\label{sec:emptyKey}
Since the data stream is unordered and unbounded, the existence of user contributions within each micro-batch can be arbitrary, as shown in Figure~\ref{fig:zero-key}. It is possible that some keys do not have any user records in a micro-batch. In the traditional streaming systems, states are not updated unless new records appear in the micro-batch\cite{to2018survey}. Therefore, the system only needs to load states with keys from the current micro-batch.

\begin{sloppypar}
However, the streaming key selection algorithm (Algorithm~\ref{Alg:key-selection}) requires us to perform the thresholding test for the entire key space, and it is possible for a key to be selected without new records. An naive approach is to load all the keys with their associated states from the system state store, and run Algorithm~\ref{Alg:key-selection} directly. Unfortunately, when the key space is large, the I/O cost and memory cost of loading the entire state table is too high. Here, we propose the empty key release prediction algorithm, together with other operational strategies, to solve the scalability challenge.
\end{sloppypar}

\begin{algorithm}[ht]
\caption{Empty Key Release Prediction}
\begin{algorithmic}[1]
\REQUIRE Data stream: $D_w=\{d_1,\ldots,d_n\}$, where each $d_i$ arrive at time $t_i$ within event-time window $w=(t_s, t_e)$, $t_s < t_i < t_e$. Triggering timestamps $\tr_w = [tr_{1}, tr_{2},...,tr_{T}]$. At each triggering timestamp $tr_{i} \subseteq \mathbb{R}$, only a sub-stream $D_t \in D_w$ is available. privacy parameters $\epsilon,\delta > 0$. Let j be the current round of key selection.
\STATE $\Delta D_{tr_{j}} \leftarrow D_{tr_{j}} - D_{tr_{j-1}}$ (which is the data stream for the micro-batch at $tr_{j}$.
\STATE $\calS_{tr_{j}}\leftarrow$ All $\recordkey \in \Delta D_{tr_{j}}$ selected via Algorithm~\ref{Alg:key-selection} at time $tr_{j}$.
\FOR{$\recordkey \in \calS_{tr_{j}}$}
    \FOR{$tr_{p}$ \textbf{from} $tr_{j+1}$ \textbf{to} $tr_{|\tr_w|}$}
            \STATE $D_{tr_{p}} \leftarrow D_{tr_{j}}$, which mimic the data stream at the next triggering timestamp $tr_{p}$, without any new record.
            \STATE Perform streaming private key selection on $D_{tr_{p}}$ for $\recordkey$ via Algorithm~\ref{Alg:key-selection} at time $tr_{p}$.
            \STATE {\bf if} $\recordkey$ is selected, {\bf then} write $\left(\recordkey, tr_p\right)$ to state store and {\bf break}.
        \ENDFOR
\ENDFOR
\end{algorithmic}
\label{Alg:empty-key-prediction}
\end{algorithm}

\mypar{Algorithm description for empty key release prediction} There are two scenarios that may trigger a key being selected: key is selected due to additional user contributions, or key is selected due to noise addition without user contributions. The first scenario is naturally handled by the streaming system, when processing micro-batch with new records. The secondary scenario is handled by Algorithm~\ref{Alg:empty-key-prediction}.

When a micro-batch $\Delta D_{tr_{j}}$ contains key $k$, the streaming private key selection algorithm for key $k$ is applied to the sub-stream $D_{tr_{j}}$. In case $k$ is not selected, we will simulate streaming private key selection algorithm executions from $tr_{j+1}$ to $tr_{|\tr|}$, using sub-stream $D_{tr_{j}}$, and predict if any future release is possible by adding leaf node with~\emph{zero} count. The predicted releasing time is $tr_{p}$, and it is written to the system state store.

After making a release prediction, the DP-SQLP will continue to process the next micro-batch. For key $k$ with predicted releasing time $tr_{p}$, there are two cases:

\begin{figure}[ht]
\centering
\includegraphics[width=0.3\textwidth]{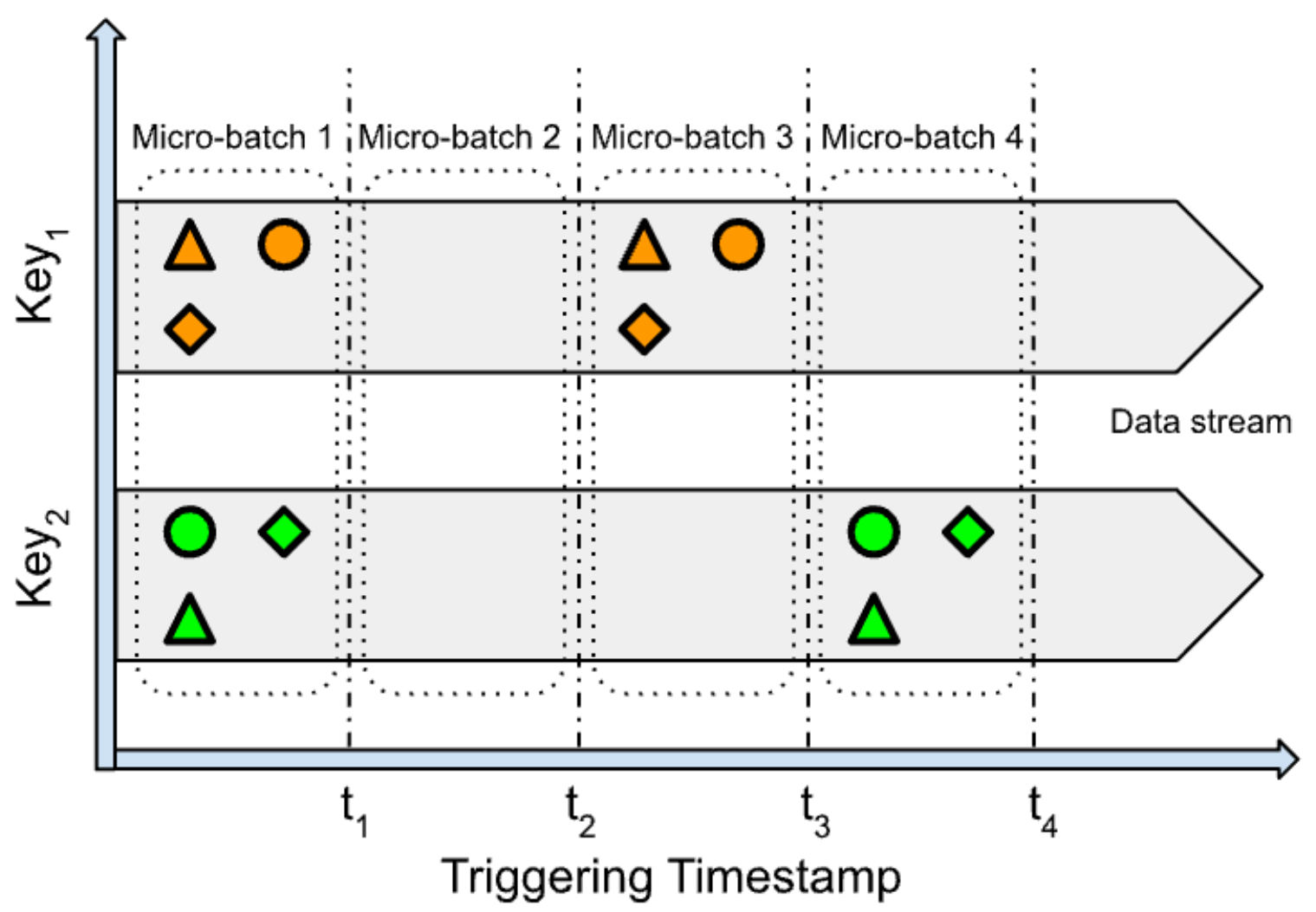}
\caption{Some key may not have records within certain micro-batch}
\label{fig:zero-key}
\end{figure}

\begin{enumerate}
    \item $k$ appears in another micro-batch $\Delta D_{tr_{n}}$ before the predicted triggering timestamp ($j < n < p$). In this case, the prior prediction result is discarded. We will perform key selection algorithm for the micro-batch $\Delta D_{tr_{n}}$. All the thresholding test for micro-batches that do not have $k$ between $tr_{j+1}$ and $tr_{n-1}$ have been performed during the prediction phase in the prior micro-batch $D_{tr_{j}}$. In addition, we will make a new prediction for $\Delta D_{tr_{n}}$.
    \item $k$ appears in another micro-batch $\Delta D_{tr_{n}}$ after the predicted triggering timestamp ($p < n$). In this case, DP-SQLP loaded system states for $k$ at the predicted time $tr_{p}$ and released data from the data buffer. We will start a new round of streaming key selection from micro-batch $\Delta D_{tr_{n}}$.
\end{enumerate}
Within these operations, some computations might be wasted (e.g., Case 1). However, when the key space is large, the reductions in I/O cost and memory cost bring in more benefits than the CPU overhead.

The prediction result is stored in the system state. The step to load predicted results is shown in Figure \ref{fig:sdp-implementation}. We also add a secondary index on the predicted timestamp to improve the state loading speed.

\subsection{Privacy Accounting for DP-SQLP}
\label{sec:accounting}
In DP-SQLP, the privacy costs occur in streaming key selection (Algorithm~\ref{Alg:key-selection}) and hierarchical perturbation (Algorithm~\ref{Alg:dptree-aggregation}). Because each user is allowed to contribute at most $C$ records, we use the combination of composition and sensitivity in privacy accounting.

\begin{itemize}
    \item \emph{Privacy accounting for streaming key selection:} When executing Algorithm~\ref{Alg:key-selection} in DP-SQLP, the \emph{value} added to each leaf node is the \emph{unique user count}. Therefore, the per-user sensitivity for each \DpTree is one. In addition, we restart the Algorithm~\ref{Alg:key-selection} once a key is selected, and data accumulated in the data buffer is released immediately. As a result, each user may participate in at most $C$ rounds of key selection. Given the $(\epsilon,\delta)$ privacy budget for each round of Algorithm~\ref{Alg:key-selection}, the total privacy cost for streaming key selection is calculated using the empirically tighter variant of advanced composition~\cite{dwork2014algorithmic} with $C$-fold.

    \item \emph{Privacy accounting for hierarchical perturbation:} For each user, in the worst case, all $C$ contributions can go to the same node of a single \DpTree, we scale up the sensitivity corresponding to any single node in the tree to $L_1=C\cdot L$, and ensure that each tree still ensures $\rho$-zCDP. After that, we use the conversion from~\cite[Proposition 3]{mironov2017renyi} to translate privacy cost from $\rho$-zCDP to $(\epsilon,\delta)$-DP guarantee~\footnote{The exact computation is from~\url{https://github.com/IBM/discrete-gaussian-differential-privacy/blob/master/cdp2adp.py\#L123}}.
\end{itemize}

Finally, the privacy costs of key selection and hierarchical perturbation are combined via advanced composition.

\section{Experiments}
\label{sec:experiments}

The experiments are performed using both synthetic and real-world data to demonstrate the data utility and scalability. The streaming DP mechanism is implemented in the DP-SQLP operator, as described in section~\ref{sec:streamingdpimp}.

\mypar{Baselines} We compare DP-SQLP with two baseline approaches for data utility.

\begin{itemize}
    \item Baseline 1 - Repeated differential privacy query. Most of the existing DP mechanisms does not have capacity to track user contributions across multiple queries. Therefore, when handling data streams, one common workaround is to repeatedly apply the one-shot DP query to the growing data set in order to get the histogram update. Thus, the overall privacy budget usage is the composition of all queries.
    \item Baseline 2 - Incremental differential privacy processing. The one-shot differential privacy algorithm is applied separately to each micro-batch, and we can get the final result by aggregating the outputs from each micro-batch. Compared with baseline 1, baseline 2 requires a similar global user contribution bounding system as DP-SQLP.
\end{itemize}
For both baseline 1 and baseline 2, the one-shot differential privacy mechanism is executed by Plume \cite{amin2022plume} with the Gaussian mechanism. Each one-shot differential privacy execution guarantees $(\epsilon, \delta)$-differential privacy. We also adopt the optimal composition theorem for DP \cite{KOV17} to maximize the baseline performance.

\mypar{Metrics} The data utility is evaluated based on 4 metrics calculated between the ground truth histogram and the differentially private histogram. 1)  Number of retained keys, which reflects how many keys are discovered during key selection process. It is also known as the $\ell_0$ norm, 2) $\ell_\infty$-norm, which reflects the worst case error: $\max\limits_{k \in\texttt{key space}} (|\hat M_k - M_k|)$, 3) $\ell_1$-norm, which reflects the worst case error: $\sum\limits_{k \in\texttt{ key space}} (|\hat M_k - M_k|)$, and 4) $\ell_2$-norm: $\sqrt{\sum\limits_{k \in\texttt{key space}} (|\hat M_k - M_k|^2)}$. 

% \begin{itemize}
%     \item Number of retained keys, which reflects how many keys are discovered during key selection process. It is also known as the $\ell_0$ norm.
%     \item $\ell_\infty$ norm, which reflects the worst case error
%     $$\ell_\infty = \max\limits_{k \in\texttt{key space}} (|\hat M_k - M_k|).$$
%     \item $\ell_1$ norm, which is an aggregated error 
%     $$\ell_1 = \sum\limits_{k \in\texttt{ key space}} (|\hat M_k - M_k|).$$
%     \item $\ell_2$ norm, also known as the euclidean norm 
%     $$\ell_2 = \sqrt{\sum\limits_{k \in\texttt{key space}} (|\hat M_k - M_k|^2)}.$$
% \end{itemize}

We choose $\epsilon = 6$ and $\delta = 10^{-9}$ as the overall privacy budget for all experiments. Within DP-SQLP, the privacy budget used by the aggregation column is $\epsilon_m=\epsilon/2$, $\delta_m=\delta/3$, and the ones used by key selection is $\epsilon_k=\epsilon/2$, $\delta_k=\delta\times2/3$. The parameter \textit{C} is chosen based on the dataset property. In this experiment, we sampled $10\%$ of one day's data and set $C$ according to the $99$ percentile of per user number of records. There are more discussions on choosing $C$ in Section~\ref{sec-param-tuning}.

For both synthetic data and real-world data, We assume the dataset represents a data stream within one day. We also shuffle users' records so that they are randomly distributed within the day. In experiments, the event-time window is also fixed to one day.

In addition to data utility, we also report the performance latency under various micro-batch sizes and number of workers.

\subsection{Synthetic Data}

We used the similar approach as \cite{amin2022plume} to generate the synthetic data to capture the long-tailed nature of real-world data. There are 10 millions unique users in the synthetic dataset. Each user draws a number of contributed records from a distribution with range $[1, 10^5]$ and mean 10 according to a Zipf-Mandelbrot distribution. The parameters\footnote{In Zipf-Mandelbrot, the sampling probability is proportional to $(x+q)^{-s}$, where $q=26$, $s=6.738$.} are chosen so that roughly $15\%$ users contribute to more than 10 records. The key in each record is also sampled from a set of size $10^6$, following a Zipf-Mandelbrot distribution\footnote{$q=1000$, $s=1.4$.}. This implies that roughly $1/3$ of records have the first $10^3$ keys.

The histogram query task we perform is a simple count query. 

\begin{lstlisting}[
           language=SQL,
           showspaces=false,
           basicstyle=\ttfamily,
           numberstyle=\tiny,
           commentstyle=\color{gray},
           caption=Histogram query for synthetic data,
           captionpos=b,
           label=synthetic-sql,
           frame=single
        ]
SELECT key, COUNT(*)
FROM SyntheticDataset
GROUP BY key
\end{lstlisting}

The per-record clamping limit $L=1$ since the aggregation function is \textit{COUNT}. We set $C=32$.

All measurements are averaged across 3 runs. The experiments are performed with 100 micro-batches and 1000 micro-batches. Within one day, 100 micro-batches correspond to roughly 15 minute triggering intervals, and 1000 micro-batches correspond to roughly 1.5 min triggering intervals. 

\begin{table}[h!]
\caption{Data utility measure with synthetic data $(\epsilon=6, \delta=10^{-9})$}
\centering
\begin{tabular}{|l|r|r|r|}
\hline
 \multirow{2}{*}{Metrics} & \multicolumn{3}{c|}{100 Micro-batches} \\
 \cline{2-4}
 & DP-SQLP & Baseline 1 & Baseline 2 \\
 \hline
 Keys & 28,338 & 435 & 191\\
 $\ell_\infty$ Norm & 1,391 & 18,077 & 21,913\\
 $\ell_1$ Norm & 17,741,225 & 50,835,203 & 58,551,587\\
 $\ell_2$ Norm & 50,039 & 430,547 & 576,425\\
 \hline
\end{tabular}

\vspace{2mm}
\begin{tabular}{|l|r|r|r|}
\hline
 \multirow{2}{*}{Metrics} & \multicolumn{3}{c|}{1000 Micro-batches} \\
 \cline{2-4}
 & DP-SQLP & Baseline 1 & Baseline 2 \\
 \hline
 Keys & 22,280 & 0 & 0\\
 $\ell_\infty$ Norm & 1,563 & 25,497 & 25,497\\
 $\ell_1$ Norm & 19,395,721 & 59,052,062 & 59,052,062\\
 $\ell_2$ Norm & 58,237 & 594,382 & 594,382\\
 \hline
\end{tabular}
\label{table:syntheticU10M}
\end{table}

The results are shown in Table~\ref{table:syntheticU10M}. There are significant data utility improvements comparing DP-SQLP with two baselines. With 100 micro-batches, the number of retained keys is increased by \textbf{65} times; the worst case error is reduced by \textbf{92\%}; the $\ell_1$ norm is reduced by \textbf{65.1\%} and the $\ell_2$ norm is reduced by \textbf{88.4\%}. The utility improvement is even more significant with 1000 micro-batches. The number of retained keys is increased from 0 to 22,280; the worst case error is reduced by \textbf{93.9\%}; the $\ell_1$ norm is reduced by \textbf{67.2\%} and the $\ell_2$ norm is reduced by \textbf{90.2\%}.

Another observation we have for DP-SQLP is its stability when the number of micro-batches increases. The utility of one-shot differential privacy mechanisms in baseline 1 and baseline 2 degrade quickly, due to the privacy budget split (baseline 1) and data stream split (baseline 2). This degradation sometimes is not linear. The number of retained keys in baseline 1 and 2 is reduced to 0 when the number of micro-batches grows from 100 to 1000. On the contrary, the utility degradation for DP-SQLP is not as significant. Indeed, when the number of micro-batches increases from 100 to 1000, for DP-SQLP, the number of retained keys is reduced by 21\%, the worst case error is increased by 12\%, the $\ell_1$ norm is increased by 9\%, and the $\ell_2$ norm is increased by 16\%.

In summary, DP-SQLP shows a significant utility improvement over one-shot differential privacy mechanisms when continuously generating DP histograms.

\subsection{Reddit Data}

%\begin{lstlisting}[
%           language=SQL,
%           showspaces=false,
%           basicstyle=\ttfamily,
%           numberstyle=\tiny,
%           commentstyle=\color{gray},
%           caption=Histogram query for the Reddit data,
%           captionpos=b,
%           label=reddit-sql,
%           frame=single
%        ]
%SELECT subreddit, COUNT(*)
%FROM RedditData
%GROUP BY subreddit
%\end{lstlisting}

In the next step, we apply the same experiment to real-world data. Webis-tldr-17-corpus~\cite{syed2017webis} is a popular dataset consisting of 3.8 million posts associated with 1.4 million users on the discussion website Reddit. Our task is to count the user participation per subreddit (specific interest group on Reddit).

We set $C=17$ and the rest of the experiment settings are the same as the synthetic data. All measurements are averaged across 3 runs.

\begin{table}[h!]
\caption{Data utility measure with the Reddit data $(\epsilon=6, \delta=10^{-9})$}
\centering
\begin{tabular}{ |l|r|r|r|  }
 \hline
 \multirow{2}{*}{Metrics} & \multicolumn{3}{c|}{100 Micro-batches} \\
 \cline{2-4}
 & DP-SQLP & Baseline 1 & Baseline 2 \\
 \hline
 Keys & 1,473 & 32 & 63\\
 $\ell_\infty$ Norm & 102,250 & 267,147 & 103,546\\
 $\ell_1$ Norm & 989,249 & 2,721,349 & 2,376,937\\
 $\ell_2$ Norm & 127,721 & 322,739 & 156,472\\
 \hline
\end{tabular}

\vspace{2mm}
\begin{tabular}{ |l|r|r|r|  }
 \hline
 \multirow{2}{*}{Metrics} & \multicolumn{3}{c|}{1000 Micro-batches} \\
 \cline{2-4}
 & DP-SQLP & Baseline 1 & Baseline 2 \\
 \hline
 Keys & 1,181 & 9 & 3\\
 $\ell_\infty$ Norm & 102,218 & 266,391 & 108,124\\
 $\ell_1$ Norm & 1,074,724 & 3,081,542 & 3,074,655\\
 $\ell_2$ Norm & 127,830 & 341,557 & 242,482\\
 \hline
\end{tabular}

\label{table:reddit}
\end{table}

The results are summarized in Table~\ref{table:reddit}, and we have similar observations as in the synthetic data experiments. DP-SQLP demonstrates significant utility improvements in the number of retained keys, $\ell_1$ norm and $\ell_2$ norm, as well as the performance stability when the number of micro-batches grows from 100 to 1000.

\begin{figure}[ht]
\centering
\includegraphics[width=0.3\textwidth]{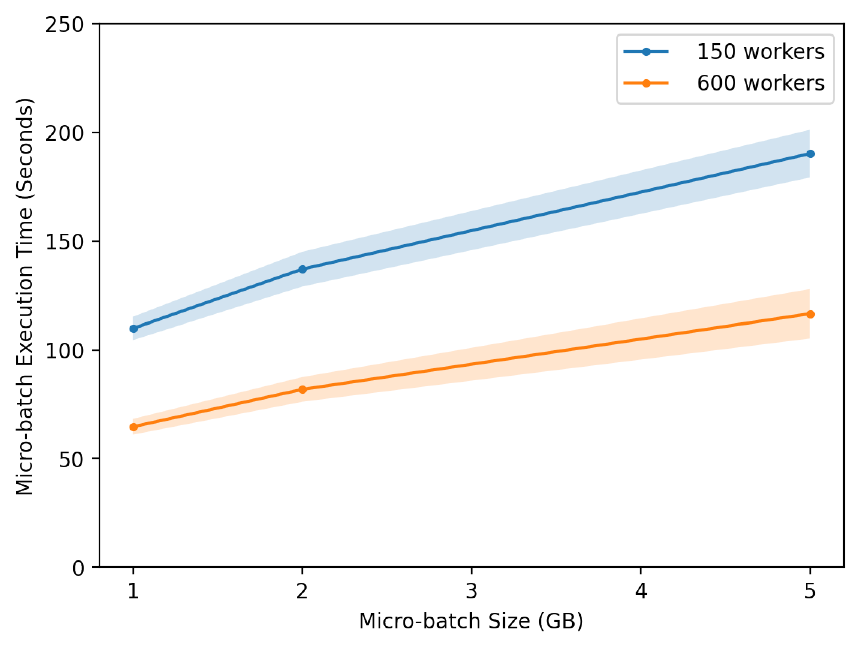}
\caption{Micro-batch execution time $(\epsilon=6, \delta=10^{-9})$}
\label{fig:latency}
\end{figure}

\subsection{Execution Performance}
The end-to-end latency of a record consists of framework latency and micro-batch execution latency. The former is determined by the streaming framework. The latter is a critical indicator for the system scalability. In this section, we report the execution latency for each micro-batch, under different micro-batch sizes and number of workers.

The results are shown in Figure~\ref{fig:latency}. All measurements are averaged across 2 runs, with shaded regions representing standard error. The execution latency grows sub-linearly as the micro-batch size increases. For example, with 150 workers, the execution latency grows 1.7 times while the data size increases 5 times from 1 GB to 5 GB. 

Figure~\ref{fig:latency} also demonstrates horizontal scalability by trading machine resources with latency. When the total number of workers increases from 150 to 600, the execution latency is reduced by 41\%, 40\%, and 39\% respectively for 1, 2, and 5 GB micro-batches.

To further test the scalability in terms of the size of key space, we generate another large synthetic dataset with 1 billion users. Each user draws the number of contributed records following Zipf-Mandelbrot distribution\footnote{$q=26$, $s=6.738$.}, generating 6 billion records in total. The key in each record is sampled from 1 billion keys following uniform distribution. When setting the micro-batch size equals 1GB and using 5500 workers, the average execution latency is 306 seconds. DP-SQLP easily handles the large load without incurring any significant performance hit.

\subsection{Parameter Tuning}
\label{sec-param-tuning}

Tuning user contribution bounding is critical to achieve good data utility. If $C$ is too small, lots of user records may be dropped due to user contribution bounding, which will lead to large histogram error. However, if $C$ is very large, the noise and key selection threshold are scaled up accordingly. Choosing the right $C$ is an optimization task.

In this section, we run the DP-SQLP with synthetic data using variable $C$ from $1$ to $50$. Figure~\ref{fig:param-tune} shows how the change of $C$ affects the number of retained keys, $\ell_1$ norm, $\ell_2$ norm and $\ell_\infty$ norm. The optimal value for $C$ (naturally) varies under different metrics. For example, the optimal $C$ for $\ell_1$ norm is around 25 whereas the optimal $C$ for $\ell_2$ norm is around 30. In comparison, the optimal $C$ for the number of retained key is around 5. Therefore, the optimal value of $C$ should be chosen according to the metric we care most about (e.g., $\ell_2$ norm). In real applications, we could use the $P99$ percentile value or DP $P99$ percentile value from a data sample as the starting point and perform a few rounds of tests to search for the optimal point.

\begin{figure}[ht]
\centering
\includegraphics[width=0.47\textwidth]{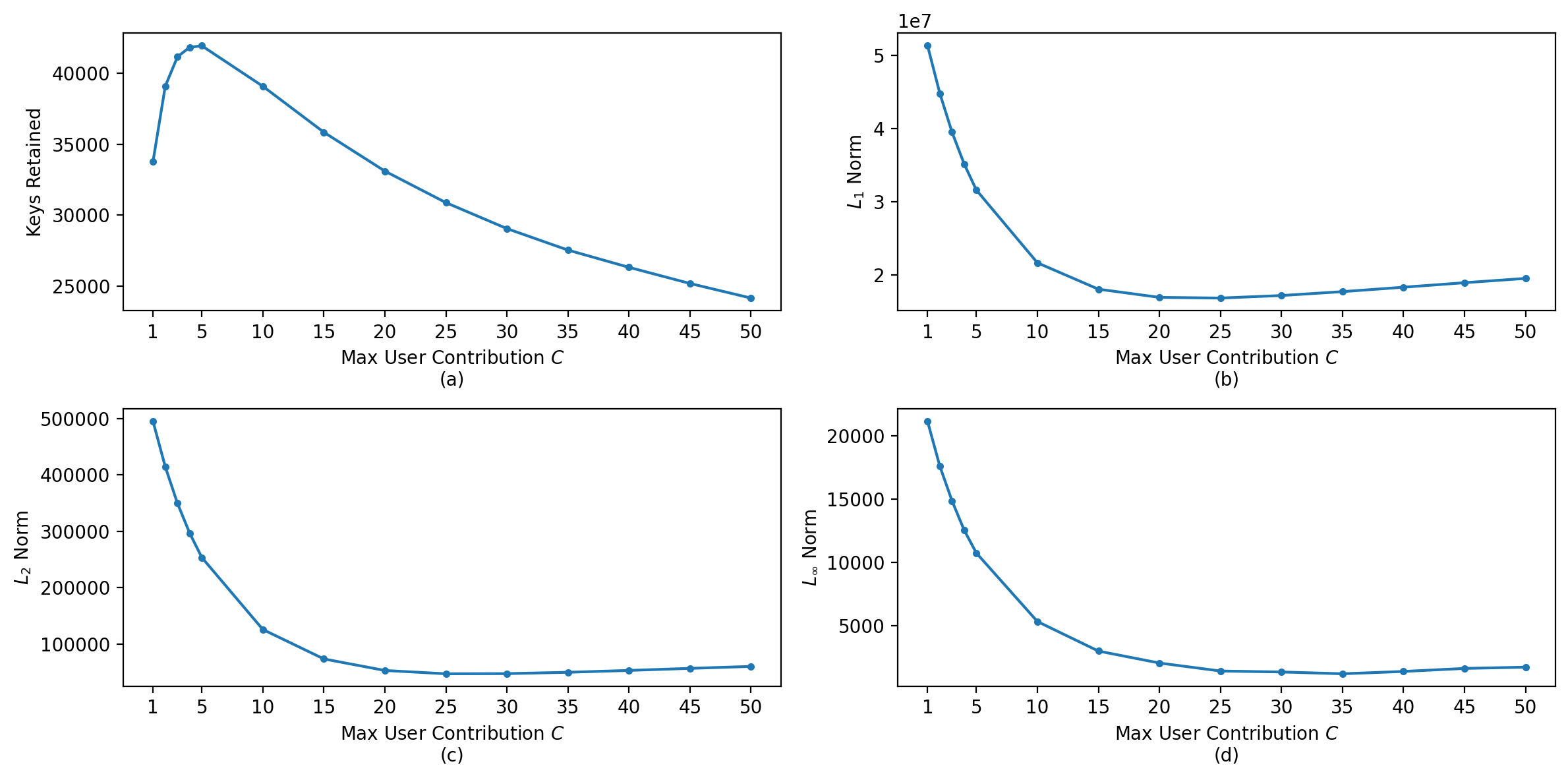}
\caption{Metrics under different contribution limit $(\epsilon=6, \delta=10^{-9})$}
\label{fig:param-tune}
\end{figure}

\section{Industry Applications}
\label{sec:application}
\subsection{Apply DP-SQLP to Google Shopping}
\label{sec:shoppingcase}
We implemented a streaming differentially private user impressions for Google Shopping. A DP-SQLP pipeline was deployed that outputs differentially private product page-view counts. The page-view was used by some shopping systems as a signal for prioritizing the crawling of pages to update product price and availability information. When it comes to use cases such as price and availability, data freshness is critical for a good user experience. Long pipeline latency in case of using batch process, may incur obsolete search results. Therefore, a streaming pipeline with low latency is required by Google Shopping.

%Privacy is a cornerstone when designing page-view. User's sensitive view activity may be revealed through the tail of distribution. Streaming differential privacy was adopted through DP-SQLP to provide a rigorous privacy guarantee.

To maintain a low latency stream of page-view counts, while addressing the privacy risks, Google Shopping set up a DP-SQLP pipeline that outputs a histogram of view count per merchant product page. The DP-SQLP pipeline processes around 20GB/s data stream, and maintains a 20 mins end-to-end latency.

In terms of data utility after adopting DP-SQLP, we were able to retain 59\% of the page-view. When focusing on the head pages, utility increases to 75\% for pages with an average view rate of 1 view/hour, and to 99.9\% for pages with an average view rate of 60 views/hour. When comparing noised impression counts with the raw counts, the relative error is around 11\%. Each user is bound to contribute one event per day to ensure user level DP guarantee, per day. We use $\epsilon=1$ for streaming aggregation. For streaming key selection, when choosing $\epsilon=1$, the equivalent threshold is around 120; when choosing $\epsilon=3$, the equivalent threshold is around 43; and when choosing $\epsilon=10$, the equivalent threshold is around 16. The overall $\delta=10^{-9}$. The max triggering window size is 150.

\subsection{Streaming DP in Google Trends}
\label{sec:trends}
Google Trends allows users to analyze the interest of search queries. The streaming private key selection (algorithm~\ref{Alg:key-selection}) is applied to it for selecting common Google Search queries with differential privacy in a streaming manner. Only queries chosen with differential privacy guarantee are shown on Google Trends website (e.g. as trending queries or related queries). The streaming pipeline ensures the 15 min end to end latency for highly searched queries to be selected with DP. Each user is bounded to contribute one event per query. We use $\epsilon=2$ and $\delta=10^{-10}$ for a user-query level DP guarantee, i.e. each query has a budget $\epsilon=2$ and $\delta=10^{-10}$. In addition, a deterministic pre-threshold $\mu=50$ is used, to yield a stronger privacy protection. That means that DP streaming selection is applied only for queries with at least 50 unique users. 
\section{Conclusion and Future Work}
\label{sec:conclusion}

In this paper, we presented a streaming differentially private system (DP-SQLP) that is designed to continuously release DP histograms. We provide a formal $(\epsilon, \delta)$-user level DP guarantee for arbitrary data streams. In addition to the algorithmic design, we implemented our system using a streaming framework similar to Spark streaming, Spanner database, and F1 query engine from Google. The experiments were conducted using both synthetic data and Reddit data. We compared DP-SQLP with two baselines, and the results demonstrated a significant performance improvement in terms of data utility. In the end, we present two industry applications that apply DP-SQLP and the streaming private key selection algorithm to the production use cases.

There are three main ways in which our system can be further extended.
First, in the design of the system we have used DP-tree aggregation~\cite{CSS11-continual,Dwork-continual,honaker2015efficient} as the baseline DP algorithm. In recent research~\cite{henzinger2023almost,denisov2022improved}, it has been shown that DP-tree aggregation is (significantly) sub-optimal in terms of privacy/utility trade-off, compared to general matrix factorization based mechanisms (DP-MF)~\cite{matousek2014factorization}. However, DP-MF algorithms are not in general compatible with systems operating on data streams. Recently, following to our work,~\cite{mcmahan2024efficient} provided a streaming variant of DP-MF. In future incarnations of our system, we plan to incorporate this.

Second, our algorithms are primarily designed to provide a centralized DP guarantee, where the final outcome of the system is guaranteed to be DP. It is worth exploring DP streaming system designs that allow stronger privacy guarantees like pan-privacy~\cite{dwork2010pan}.

Third, we bound the contribution of each user globally by $C$. However, for higher fidelity, it is important to explore approaches to perform per-key contribution bounding. Naive approaches that address this issue can get complicated due to the fact that we are dealing with an input driven stream.
%\section*{Acknowledgement}
%\label{sec:acknowledgement}
\begin{acks}
We would like to thank Olaf Bachmann, Wei Hong, Jason Peasgood, Algis Rudys, Daniel Simmons-Marengo, Yurii Sushko and Sergei Vassilvitskii for the discussions and support to this project.
\end{acks}

\bibliographystyle{ACM-Reference-Format}
\bibliography{reference}
\appendix
\section{Terms from Streaming Data Processing}
\label{sec:glossary}

\begin{definition}[Input driven stream processing]
For a stream processing system, the statistic release is driven by inputs. The output computation is based on the current input (stateless), or a sequence of inputs (stateful)~\cite{to2018survey}.

\label{def:input-driven}
\end{definition}

\begin{definition}[Event time and event time domain]
Event time is the time at which the event itself actually occurred, i.e. a record of system clock
time (for whatever system generated the event) at the time of occurrence. Each event timestamp is determined from a discrete, ordered domain, which is called event time domain~\cite{srivastava2004flexible}.
\label{def:event-time}
\end{definition}

\begin{definition}[Processing time and processing time domain]
Processing time is the time at which an event is observed by the stream processing system, i.e. the current time according to the system clock. Each processing timestamp is determined from a discrete, ordered domain, which is called processing time domain.
\label{def:processing-time}
\end{definition}

\begin{definition}[Windowing]
Windowing slices up a dataset into finite chunks for processing as a group. It determines where in event time data are grouped together for processing~\cite{akidau2015dataflow}. Some common windows include fixed windows, sliding windows and sessions.

\label{def:windowing}
\end{definition}

\begin{definition}[Triggering]
Trigger is a mechanism to stimulate the output of a specific window at a grouping operation. It determines when in processing time the results of groupings are emitted as panes~\cite{akidau2015dataflow}. There are many ways to define a trigger, including watermark, processing time timers and data arrival.

\label{def:triggering-time}
\end{definition}

\begin{definition}[Micro-batch]
Each micro-batch is a partitioned dataset. The continuous data stream is discretized into many micro-batches~\cite{akidau2015dataflow}, and the execution of each micro-batch resembles continuous processing.

\label{def:microbatch}
\end{definition}

\section{Differential Privacy on Streams}
\label{sec:dp-on-stream-appendix}

Consider a data stream which is a collection of records $D=[d_1, d_2,...)$, where each $d_i$ is a tuple $(\recordkey,$ $\recordvalue,$ $\recordtimestamp,$ $\recorduserid)$ that can be assigned to certain event-time window $w \in W$. The computation of records in $w$ happens at trigger time $tr_i \in \tr_{w}$.

Throughout this paper we will consider algorithms which take streams as input and produce output in trigger times from $\tr$ of certain window $w$: at each trigger time 
$tr_i$, the algorithm processes data up to this time. \emph{While the discussion of DP processing focus on each window, the DP guarantee is for the entire data stream, because user contributions are bounded over the complete stream.}

In the following, we provide a formal definition of differential privacy we adhere to in the paper (which includes the notion of neighboring data streams). Note that by Definition~\ref{def:neighbors} we consider user level differential privacy with fixed triggering times. Triggering times are not protected.

\begin{defn}
\label{def:neighbors}
Streams $D, D'$ are called neighbouring if they only differ by the  absence or presence of the data of one user while keeping the trigger times $\tr$ fixed between the two datasets.
\end{defn}

\begin{defn}[Differential Privacy~\cite{DMNS,ODO}] 
An algorithm $\calA$ is $(\epsilon,\delta)$-differentially private (DP) if for any neighboring pairs of data streams $D$ and $D'$ the following is true for any $S\subseteq\texttt{Range}(\calA)$:
$$\Pr[\calA(D)\in S]\leq e^\epsilon\Pr[\calA(D')\in S]+\delta.$$
\label{def:DP}
\end{defn}
We can similarly define Concentrated DP (zCDP) \cite{dwork2016concentrated,bun2016concentrated} with the neighbouring relation given by Definition \ref{def:neighbors}. We use zCDP as an analysis tool\footnote{Throughout the paper, zCDP is used as a convenient tool to analyze DP-Binary Tree. The privacy guarantee is reported in $(\epsilon,\delta)$-DP.}, as it has clean composition properties, but we convert the final guarantee into $(\varepsilon,\delta)$-DP using standard tools.
We use three facts:
\begin{itemize}
    \item Given a query $q$ with sensitivity $\Delta_2 := \sup_{x,x'\atop\text{neighboring}} \|q(x)-q(x')\|_2$, releasing $M(x) = \mathcal{N}(q(x),\sigma^2\mathbb{I})$ satisfies $\frac{\Delta_2^2}{2\sigma^2}$-zCDP.
    \item Composing an algorithm satisfying $\rho_1$-zCDP with an algorithm satisfying $\rho_2$-zCDP yields an algorithm satisfying $(\rho_1+\rho_2)$-zCDP.
    \item For any $\rho,\delta > 0$, $\rho$-zCDP implies $\left(\rho + 2\sqrt{\rho\log(1/\delta)},\delta\right)$-DP.
\end{itemize}

\section{Privacy Under Continual Observation and Binary Tree Aggregation} 
\label{sec:dp-tree-appendix}

Consider a data stream  $D = x_1,\ldots,x_n$ with each $x_i\in[0,L]$. The objective is to output $\left\{S_j=\sum\limits_{i=1}^j x_i\right\}$ while preserving DP (with the neighborhood relation defined w.r.t. changing any one of the $x_i$'s). We will use the so-called \emph{binary tree aggregation algorithm}~\cite{Dwork-continual,CSS11-continual,honaker2015efficient} stated below.

\begin{algorithm}[ht]
\caption{DP-Binary Tree Aggregation (exposition from~\cite{kairouz2021practical})}
\begin{algorithmic}[1]
\REQUIRE Data set: $D=\{x_1,\ldots,x_n\}$ with each $x_i\in[0,L]$, noise standard deviation $\sigma$.
\STATE Initialize a complete binary tree with $2^{\lceil\lg(n)\rceil}$ leaf nodes, with each node being sampled from $\calN(0,\sigma^2)$.
\FOR{$i\in[n]$}
    \STATE Add $x_i$ to all the nodes on the path from the $i$-th leaf node to the root of the tree.
    \STATE Initialize $S^{\sf priv}_i\leftarrow 0$, and convert $i$ to the binary representation $[b_1,\ldots,b_h]$ with $b_1$ being the most significant bit, and $h=\lceil\lg(n)\rceil$ being the height of the tree.
    \STATE Let $[\nd_1,\ldots,\nd_h]$ be the nodes from the root to the $i$-th leaf node of the tree.
    \FOR{$j\in[h]$}
        \STATE {\bf If} $b_j==1$, then add the value in the left sibling of $\nd_j$ to $S^{\sf priv}_i$. Here, if $\nd_j$ is the left child, then treat it as its own sibling.
    \ENDFOR
    \STATE Output $S^{\sf priv}_i$.
\ENDFOR
\end{algorithmic}
\label{Alg:abcd2}
\end{algorithm}

One example DP-Binary tree with nodes encoding is shown in Figure~\ref{fig:dp-tree}.

\begin{figure}[ht]
\centering
\includegraphics[width=0.4\textwidth]{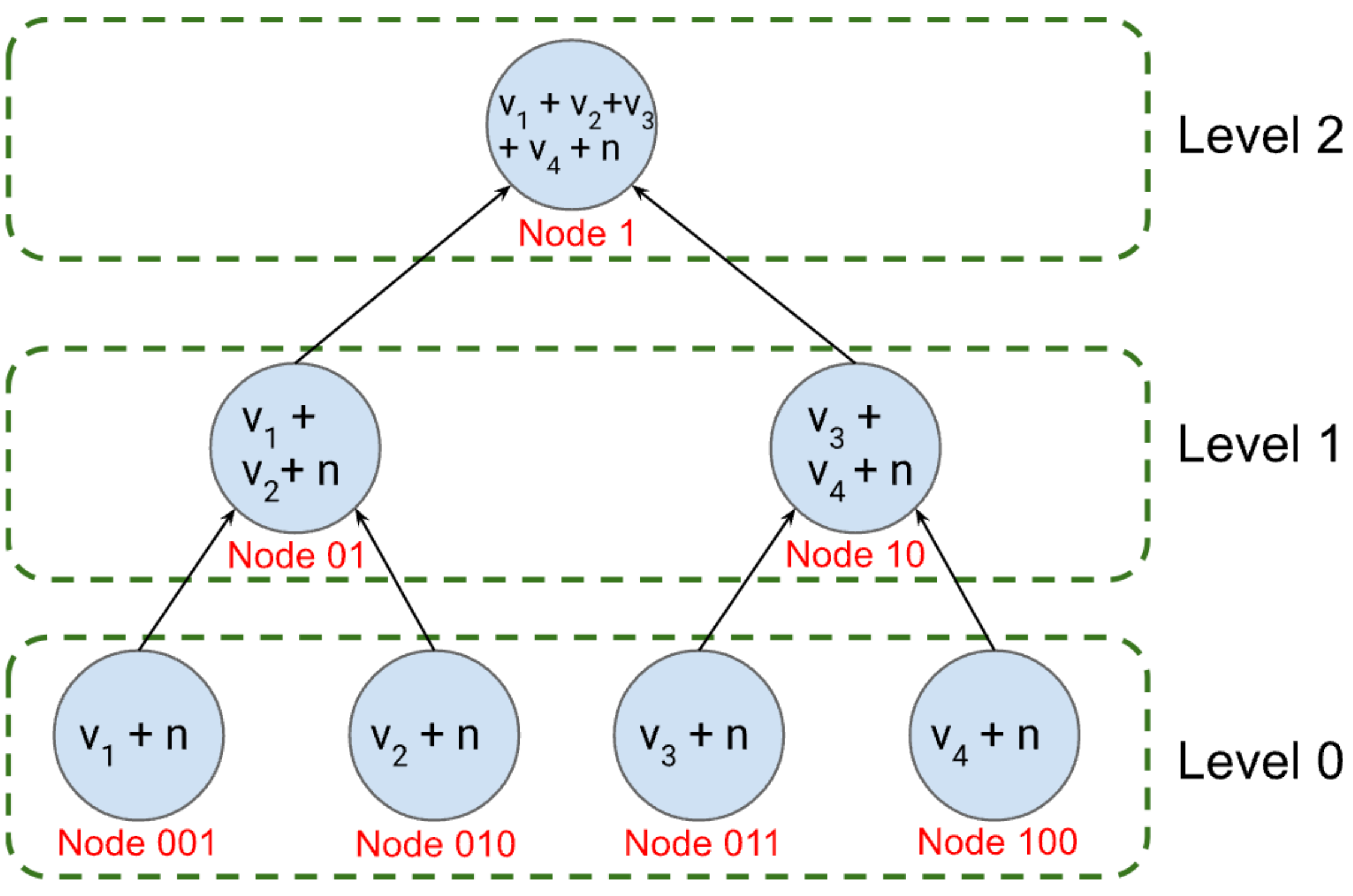}
\caption{DP-Tree with nodes encoding}
\label{fig:dp-tree}
\end{figure}

\begin{thm}
Algorithm~\ref{Alg:abcd2} satisfies $\frac{L^2\lceil\lg(n)\rceil}{2\sigma^2}$-zCDP. This is equivalent to $(\varepsilon, \delta)$-DP with $\varepsilon = \frac{L^2\lceil\lg(n)\rceil}{2\sigma^2} + 2 \sqrt{\frac{L^2\lceil\lg(n)\rceil}{2\sigma^2} \lg (1/\delta)}$.
\label{thm:privGuaranteeTree}
\end{thm}
The proof of Theorem~\ref{thm:privGuaranteeTree} follows immediately from the zCDP guarantee of Gaussian mechanism, and composition over the levels of the binary tree. One can also show the following in terms of utility.
\begin{thm}
With probability at least $1-\beta$, the following is true for any $i\in[k]$:
$$\left|S^{\sf priv}_i-S_i\right|\leq \sqrt{\frac{2\ln(n/\beta)\lceil\lg(n)\rceil\sigma^2}{\pi}}.$$
\label{thm:utilGuaranteeTree}
\end{thm}
The proof of Theorem~\ref{thm:utilGuaranteeTree} follows immediately from the tail properties of Gaussian distribution. One can improve the constants in the guarantee by using variance reduction techniques from~\cite{honaker2015efficient}, which we briefly describe below.

\mypar{Honaker estimation for variance reduction} The idea of Honaker estimation~\cite{honaker2015efficient} is to use multiple estimates for the information of same node in the binary tree for variance reduction. Consider any node $\node_{i}$, and the information that is inserted into it by Algorithm~\ref{Alg:abcd2}. Let $\calT_{\node_{i}}$ be the subtree rooted at node $\node_{i}$. Notice that sum of all the nodes at each level of $\calT_{\node_{i}}$ is an unbiased and independent estimate of the information contained in $\node_{i}$. Furthermore, we also know the variance of such an estimate. For example, if a level $\calT_{\node_{i}}$ has $m$ nodes, then the variance of the estimate from that level is $m\sigma^2$. We can use this information to reduce the variance in the estimate of $\node_{i}$ The following is the formalization of this idea.

Let $\level_0,\ldots,\level_{\kappa}$ be the levels of the subtree rooted at ${\sf node}_{i}$, with $\level_0$ being the level of ${\sf node}_{i}$. Notice that sum of all the nodes at level $\level_j$ is an unbiased estimate of the true value of node ${\sf node}_{i}$ (called the ${\sf sum}(\level_j)$), and the variance is $\sigma_j^2=2^{j}\sigma^2$ (since there are $2^{j}$ nodes at level $\level_j$). Honaker estimate is $\sum\limits_{j=0}^{\kappa-1} c_j\cdot {\sf sum}(\level_j)$, where $c_j=\frac{1/2^j}{\sum\limits_{j=0}^{\kappa-1} (1/2^{j})}$. Therefore, the variance in the Honaker's estimate is as follows:
\begin{equation}{\sf Variance}({\sf node}_{i})=\left(\sum\limits_{j=0}^{\kappa-1} c_j^2\cdot 2^j\right)\cdot\sigma^2=\frac{1}{2\cdot(1-2^{-\kappa})}\cdot \sigma^2.
\label{eq:HonackerVar}
\end{equation}

Now consider the nodes ${\sf node}_1,{\sf node}_2,\ldots,{\sf node}_m$ to be the nodes used to compute the DP variant of $S_i$ at time instance $i\in[n]$ in Algorithm~\ref{Alg:abcd2}. Notice that the noise due the Honaker estimate mentioned above is still independent for each of the nodes. Hence, we have the error $S^{\sf priv}_i-S_i$ follows the following distribution:
\begin{equation}
\left(S^{\sf priv}_i-S_i\right)\sim\mathcal{N}\left(0,\sum\limits_{j=1}^m {\sf Variance}({\sf node}_{j})\right).
\label{eq:distTree}
\end{equation}
Unlike Theorem~\ref{thm:utilGuaranteeTree} we write the exact distribution of the error above is because, we will use this distribution to decide on the threshold for key selection in Section~\ref{sec:partkeyselection}.

\mypar{Ensuring user-level DP}

Although existing works commonly assume that each user corresponds to a single record, in practice, a single user can contribute multiple records. Thus, we must account for this in the DP analysis. The simplest approach -- which we follow -- is to limit the number of contributions per person to some constant $C$ within the entire data stream $D$. If we have a DP guarantee for a single record, we can apply group privacy to obtain a DP guarantee for $C$ records. Specifically, if we have $(\varepsilon,\delta)$-DP with respect to the addition or removal of one record, then this implies $\left(C \cdot \varepsilon, \frac{e^{C \cdot \varepsilon}-1}{e^\varepsilon-1} \delta \right)$-DP with respect to the addition or removal of $C$ records \cite{vadhan2017complexity}. The group privacy is used in privacy accounting for hierarchical perturbation (Section~\ref{sec:accounting}).

In worst case scenario, we have to go through group privacy to achieve user-level DP. However, in practice, depending on how we bound user contribution, if the sensitivity of any user is inflated by at most a factor of $C$, the composition bounds instead of group privacy bounds is applicable, which gives better bounds. Later in Section~\ref{sec:partkeyselection}, we apply advanced composition to extend the privacy guarantee from the single user contribution ($C=1$) to a more general case ($C\geq1$).
\section{Missing Details from Algorithm~\ref{Alg:key-selection}}
\label{app:key-selection}

In this section, we provide the missing details for the key selection algorithm (Algorithm~\ref{Alg:key-selection}).

\mypar{Setting $\sigma$ in Line~\ref{line:keySelection1}} The binary tree has at most $T=|\tr|$ nodes. By Theorem~\ref{thm:privGuaranteeTree}, the tree aggregation algorithm satisfies $\rho=\frac{\lceil\log_2(T)\rceil}{2\sigma^2}$-zCDP. Now, we translate the zCDP guarantee into $(\epsilon,\delta)$-DP (with $\epsilon$ being a parameter of $\sigma$) via~\cite[Proposition 2.3]{mironov2017renyi}~\footnote{The exact computation is from~\url{https://github.com/IBM/discrete-gaussian-differential-privacy/blob/master/cdp2adp.py\#L123}}. Using this we obtain $\sigma$ as a function of fixed $\epsilon$ and $\delta$ and $T$, when used as parameter in Algorithm~\ref{Alg:key-selection} is guaranteed to ensure $(\epsilon,\delta)$-DP for the binary tree used there.

\mypar{Computing accuracy threshold in Line~\ref{line:thresh}} Consider the tree $\tree_k$, and nodes $\node_1,\ldots,\node_k$ (from the tree $\tree_k$) required in the computation of $\hat{q}_{{tr_{i},k}}$ By~\eqref{eq:distTree}, we have the following:
\begin{equation}
\hat{q}_{{tr_{i},k}}-\countt_k(D_{tr_{i}})\sim\calN\left(0,\underbrace{\sum\limits_{j=1}^k {\sf Variance}({\sf node}_{j})}_{\lambda^2}\right).\label{eq:treeCompute}
\end{equation}
In~\eqref{eq:treeCompute}, ${\sf Variance}({\sf node}_{j})$ is obtained by~\eqref{eq:HonackerVar}. We now compute $\tau$ based on the inverse cdf of $\calN(0,\lambda^2)$ at $1-\beta$.
\section{Proof for Theorems 3.1}
\label{sec:theorem31-proof-appendix}

\begin{proof}
   Consider two datasets $D$ and $D'$ which differ by the addition or removal of one key $k_*$.
    Let $A(D)$ denote Algorithm \ref{Alg:key-selection} run on input $D$. 
    First we note that we can ignore all other keys $k \ne k_*$ because the behaviour of Algorithm \ref{Alg:key-selection} on those keys is independent of its behaviour with respect to $k_*$.
    
    For the purposes of the analysis, we consider a different algorithm $A'(D)$ which does the following with respect to $k_*$. It initializes the tree aggregation mechanism $\tree_{k_*}$ at the beginning. At each time $t \in \tr$, the algorithm $A'(D)$ computes $\hat{q}_{t,k_*} \gets \gettotal(\tree_{k_*})$. If $\hat{q}_{t,k_*} > \mu + \tau$, then it outputs $(k_*, \hat{q}_{t,k_*})$; otherwise it outputs nothing about $k_*$ at time $t$.
    
    The algorithm $A'$ is $(\varepsilon,\delta)$-DP. This is because it is simply a postprocessing of the tree aggregation mechanism, which is set to have this DP guarantee.
    
    The only difference between the outputs of $A(D)$ and $A'(D)$ is that, if $\countt_{k_*}(D_{t_i}) \le \mu$ and $\hat{q}_{t,k_*} > \mu + \tau$, then $A'(D)$ outputs $(k_*, \hat{q}_{t,k_*})$, but $A(D)$ outputs nothing regarding $k_*$.
    Since $\hat{q}_{t,k_*} > \mu + \tau$ is meant to be an approximation to $\countt_{k_*}(D_{t_i}) \le \mu$, this means the outputs only differ when the tree aggregation mechanism $\tree_{k_*}$ has error $>\tau$. The accuracy guarantee of $\tree_{k_*}$ ensures this happens with probability at most $\beta$. That is, we can define a coupling such that $\pr{}{A(D) \ne A'(D)} \le \beta$ (and the same for $D'$ in place of $D$).
    
    Thus we can obtain the DP guarantee for $A$:
    Let $E$ be an arbitrary measurable set of outputs of $A$. We have
    \begin{align*}
        \pr{}{A(D) \in E} &\le \pr{}{A'(D) \in E} + \pr{}{A(D) \ne A'(D)} \\
        &\le e^\varepsilon \cdot \pr{}{A'(D') \in E} + \delta + \pr{}{A(D) \ne A'(D)} \\
        &\le e^\varepsilon \cdot (\pr{}{A(D') \in E} + \pr{}{A(D) \ne A'(D)}) + \delta \\ 
            &\ \ \ \ + \pr{}{A(D) \ne A'(D)} \\ 
        &= e^\varepsilon \cdot \pr{}{A(D') \in E} + \delta + (e^\varepsilon + 1) \cdot \pr{}{A(D) \ne A'(D)} \\
        &\le e^\varepsilon \cdot \pr{}{A(D') \in E} + \delta + (e^\varepsilon + 1) \cdot \beta.
    \end{align*}
\end{proof}

\section{Proof for Theorems 3.2}
\label{sec:theorem32-proof-appendix}

\begin{proof}
By using Theorem~\ref{thm:utilGuaranteeTree}, and using the notation in Line~\ref{line:thresh} the following is immediate, when the noise added to each node of the binary tree in Algorithm~\ref{Alg:key-selection} is $\calN(0,\sigma^2)$. For any fixed key $k$,

{\small
\begin{equation}
\begin{split}
    \pr{\tree_k}{\forall tr_{i}\in\tr ~~  |\hat{q}_{{tr_{i},k}}-\countt_k(D_{tr_{i}})| \le  \sqrt{\frac{2\ln(T/\beta)\lceil\lg(T)\rceil\sigma^2}{\pi}}} \\
    \ge 1-\beta.
    \label{eq:threshutil1}
\end{split}
\end{equation}}
Setting $\sigma^2=\frac{2\lceil\lg(T)\rceil\ln(1.25/\delta)}{\epsilon^2}$, and using the translation of zCDP to $(\epsilon,\delta)$-DP (based on the privacy statement in Theorem~\ref{thm:privGuaranteeTree}) completes the proof.
\end{proof}
\end{document}